\documentclass[ journal]{IEEEtran}
\usepackage{amssymb}
\usepackage{amsmath,amsfonts}
\usepackage{algorithmicx}
\usepackage{amsthm}
\usepackage{algorithm}
\usepackage{array}
\usepackage{textcomp}
\usepackage{stfloats}
\usepackage{url}
\usepackage{verbatim}
\usepackage{graphicx}
\usepackage{cite}
\usepackage{algpseudocode}
\usepackage{bm}
\usepackage{graphicx}
\usepackage{subfigure}
\usepackage{float}%提供float浮动环境
\usepackage{booktabs}%提供命令\toprule、\midrule、\bottomrule
\newcommand{\RNum}[1]{\uppercase\expandafter{\romannumeral #1\relax}}

\theoremstyle{definition}

\newtheorem{lemma}{Lemma}

\theoremstyle{remark}

\begin{document}

\title{Coverage Probability and Average Rate  Analysis of Hybrid Cellular and 
Cell-free Network 
}
%with Stochastic Geometry }

\author{{Zhuoyin Dai,  Jingran Xu, Xiaoli Xu,~\IEEEmembership{Member, IEEE},
Ruoguang Li,~\IEEEmembership{Member, IEEE}, \\ 
Yong Zeng,~\IEEEmembership{Senior Member, IEEE}, 
and Jiangbin Lyu,~\IEEEmembership{Member, IEEE}
% \\ Shi Jin, ~\IEEEmembership{Senior Member, IEEE}, 
%   Tao Jiang, ~\IEEEmembership{Fellow, IEEE}
  }  
  %{Email: \{zhuoyin\_dai, jingran\_xu, xiaolixu, ruoguangli, yong\_zeng\}@seu.edu.cn}

        % <-this % stops a space

 \thanks{
          % This work was supported by the National Key R\&D Program of China with 
          % grant number 2019YFB1803400.
          % Part of this work has been presented at 
          % the  2021 IEEE/CIC ICCC Workshops, Xiamen, China, 28-30 Jul.  2021 \cite{Dai2021a}.
          
          % This work was supported by the National Key R\&D Program of China with
          % Grant number 2019YFB1803400, by the Natural Science Foundation of China
          % under Grant 62071114, by the Fundamental Research Funds for the Central
          % Universities 2242022k30005, by the Outstanding Projects of Overseas 
          % Returned Scholars of Nanjing under Grant 1104000396, and also by
          % the China Postdoctoral Science Foundation under Grant 2022M710692.
          Part of this work has been presented   at the  
          2024 IEEE 
          %Wireless Communications and Networking Conference (WCNC) 
          WCNC Workshops, 
          Dubai, United Arab Emirates, 21-24 Apr.  2024 \cite{dai2024hybrid}.

          Z. Dai, J. Xu, X. Xu, and Y. Zeng  are with the National Mobile Communications Research Laboratory,
          %and Frontiers Science Center for Mobile Information Communication and Security
          Southeast University, Nanjing 210096, China. Y. Zeng is also with the 
          Purple Mountain Laboratories, Nanjing 211111, 
          China. (e-mail: \{zhuoyin\_dai, jingranxu, xiaolixu, yong\_zeng\}@seu.edu.cn). 
          (\emph{Corresponding author: Yong Zeng.})
          
          R. Li is with the College of Information Science and Engineering, Hohai University, Changzhou 213200, China. (e-mail: ruoguangli@hhu.edu.cn).
          
          J. Lyu is with the School of Informatics, Xiamen University,  Xiamen 361005, China. (e-mail: ljb@xmu.edu.cn).

          }

} 

 \maketitle

\begin{abstract}
  Cell-free wireless networks deploy
  distributed access points (APs) to simultaneously serve user equipments (UEs) across 
  the service region and are regarded as one of the most promising network architectural paradigms.
  Thanks to the collaboration among APs, cell-free wireless networks can provide stable and uniform communication performance 
  while enhancing both  energy and spectral efficiency.
  % As one of the most promising network architectural paradigms, 
  % cell-free wireless networks can offer stable and uniform communication performance 
  % while improving both the energy and spectral efficiency of the system by deploying 
  % distributed access points (APs) to simultaneously serve user equipments (UEs) within 
  % the network area. 
  Despite recent advances in the performance analysis and 
  optimization of cell-free wireless networks, it remains an open question whether 
  large-scale deployment of APs in existing wireless networks can cost-effectively 
  achieve communication capacity growth. Besides, the realization of a cell-free network 
  is considered to be a gradual long-term evolutionary process in which 
  cell-free APs will be incrementally introduced into existing cellular 
  networks, and   form a hybrid 
  communication network with the existing cellular base stations (BSs). 
  Such a collaboration will bridge  the gap between the established cellular network and the innovative cell-free 
  network. Therefore, hybrid cellular and cell-free networks (HCCNs) emerge as a practical and 
  feasible solution for advancing cell-free network development, and it is worthwhile 
  to further explore its performance limits.
  This paper presents a stochastic geometry-based hybrid cellular and cell-free 
  network model to analyze the   distributions of signal and 
  interference and reveal their mutual coupling. 
  Specifically, in order to benefit the UEs from both the cellular BSs and the 
  cell-free APs, a conjugate beamforming design is employed, and the aggregated signal 
  is analyzed using moment matching. Then, the coverage 
  probability of the hybrid network is characterized by deriving the Laplace 
  transforms and  their higher-order derivatives of  interference components. 
  Furthermore, the average achievable  rate of the hybrid network over channel fading 
  is derived based on  the interference coupling 
  analysis. The analytical results provide useful insights for  system evaluation and 
  are validated through 
  extensive simulations.
  
\end{abstract}

% \begin{IEEEkeywords}
   
% \end{IEEEkeywords}

\section{Introduction}
%\IEEEPARstart{T}{ith}   
Future wireless communication networks will witness a proliferation of mobile 
applications and an unprecedented increase in wireless data. 
However, achieving higher spectrum and energy efficiency at 
lower power costs remains a tough challenge in current research. 
As one of the most significant wireless technologies 
developed recently, cell-free networks
% which controls the distributed access points (APs) in the system through a 
% central processing unit (CPU), 
are regarded as a promising network architecture   for
the beyond fifth-generation (B5G) and sixth-generation (6G) mobile communication systems 
\cite{ngo2017cell,bjornson2020making}.
Unlike existing cellular systems, the cell-free network is a user-centric 
coverage architecture that eliminates the traditional concept of 
cellular boundaries \cite{buzzi2017cellfree,bjornson2020scalable }. 
Cell-free access points (APs) cooperate under the guidance of a central processing unit (CPU)
to provide uniform good services to 
user equipments (UEs) on the same time-frequency resources,
leading to a higher spatial multiplexing gain \cite{nayebi2017precoding,zhuoyindai2023characterizing}. 
Cell-free networks may improve both the energy  and spectral efficiency of the system, 
improve the likelihood  of finding nearby APs around UEs through rich macro diversity, 
and provide uniform stable services, thus effectively reducing the 
performance gap among UEs.

However, the application of cell-free systems in future   mobile networks still faces severe challenges.
First, updating the existing commercial cellular network to a cell-free system 
is expected to be a long-term  gradual process, which requires numerous expenses to deploy cell-free APs and the
corresponding forward links.
Second,   simple deployment of cell-free APs without collaboration inevitably 
leads to mutual interference with legacy cellular systems, 
even resulting in a degradation  of system performance. 
Therefore, the deployment of cell-free systems will involve a long period of 
gradual evolution, and hybrid cellular and cell-free networks (HCCNs) are both 
a necessary and desirable choice throughout  this process.

Resource allocation and performance analysis of HCCNs have been preliminarily  investigated in several existing works.
The introduction of cell-free systems for coexistence and collaboration 
in existing legacy cellular systems, along with the corresponding precoding, 
power control, etc., are outlined in \cite{kim2022howwill,zhangyu2023cellularcellfree}. 
With appropriate coordinated beamforming and UE association criteria, hybrid 
small cellular and  cell-free systems 
can provide better downlink transmission performance for both static and dynamic UEs 
than the single 
architecture \cite{Elhoushy2021towards}.
However, the above works do not consider the impact of the spatial distribution 
of system elements such as  base stations (BSs), APs, and UEs  on 
network performance.

Traditional grid-based network element deployment models struggle to reflect 
actual system performance due to the irregularity and  densification of wireless node 
distributions. 
Stochastic geometry, which uses  point processes to model the spatial distribution 
of wireless nodes, can accurately characterize the lower bound of practical 
system performance. 
Several studies have been conducted using stochastic geometry for 
channel hardening analysis \cite{chen2018channelhardening},
power control \cite{hoang2018cellfree}, 
energy efficiency analysis \cite{Papazafeiropoulos2020performance}, etc.,
for cell-free networks
and coordinated multiple points (CoMP) communication. 
However, there is  a lack of research on the   characterization 
of HCCNs. 
Furthermore, the current stochastic geometry-based heterogeneous network studies, 
which isolate different network layers from each other \cite{heath2013modeling,chun2015modeling}, 
are not applicable to the analysis of HCCNs.

To obtain important  insights into the performance limits of the hybrid network, 
this paper develops a stochastic geometry-based analytical model 
for HCCNs that reveals the coupling relationship of the signal and interference from 
both cellular BSs and cell-free APs. 
However, characterizing the distribution of the signal strength and 
signal-to-interference plus
noise ratio (SINR)
is quite difficult due to the aggregation  of the desired signals and interference
from BSs and APs.    
To address this issue, 
the statistical characteristics of the BS and  AP channels are first analyzed, 
and then the closed-form expressions for the average AP signal strength and 
interference power are derived. The aggregated signal strength distribution 
is further approximated using moment matching.
Then, the network coverage probability and average 
achievable rate
are characterized based on the Laplace transforms of the system interference elements.
The analysis of  coverage probability  and 
achievable rate
are verified by extensive simulations. 
Based on the theoretical analysis and simulation results, both the coverage probability and the average achievable rate 
of HCCN can be practically enhanced by deploying cell-free APs to collaborate with BSs.
On the one hand, APs with higher density bring higher macro diversity, which  further improves the performance of HCCN.
On the other hand, the optimal AP power is often determined by multiple network parameters such as UE density, AP density, etc., 
which give us valuable guidance for  network deployment and interference management 
towards spectrum- and  energy-efficient HCCNs.

The main contributions of this 
work are summarized as follows:
\begin{itemize}
\item
Considering the long-term evolution process of the cell-free architecture, 
the system model of HCCNs is developed, in which the spatial locations of
 multi-antenna BSs, 
multi-antenna APs, and   UEs are modeled as independent homogeneous Poisson point 
processes
(HPPP). Each UE in the HCCN is served by all the cell-free APs and the nearest cellular BS.
To explore the network coverage probability and the average achievable rate,
the   interference power in SINR is decoupled into  intra-cell and inter-cell interference,
as well as interference caused by AP signals to other UEs.
Further, the coverage probability and the average achievable rate of HCCNs are 
expressed as  functions of SINR.

\item	
%For the further derivation of the coverage probability and the average achievable rate,
The statistical channel distributions of BSs and APs are characterized.
The desired signal due to the APs is first approximated by its average. 
Considering that the desired signal follows the  distribution of the
square of the shifted Nakagami random variable, moment matching is applied to simplify 
and approximate the desired signal to a Gamma distribution.
With conjugate beamforming, intra-cell interference is approximated as a Gamma 
random variable, while inter-cell interference is approximated as the sum of 
independent Gamma variables. Further, the  interference   due to the APs is also approximated
by the corresponding average.

\item	 The coverage probability of  the  HCCN is derived based on the
distribution analysis of signal strength and interference. Specifically,
the coverage probability is the complementary
cumulative distribution function (CCDF) of the desired signal, and is 
further  transformed into the  higher-order
derivatives of the Laplace transforms of intra-cell and inter-cell interference.
On this basis, both linear weighted probability and normal distribution approximation 
are used to optimize the derivation of coverage probability.

\item The average achievable rate of the HCCN is characterized as the integral over 
the nearest BS distance. 
Considering the coupling between the desired signal and 
the interference, the average rate can be obtained directly by 
deriving the Laplace transforms of inter-cell and intra-cell interference, 
as well as the sum of intra-cell interference and the desired signal.

\end{itemize}

This paper is structured as follows:
Section II describes the system model and models  the distributions 
of network nodes 
in HCCNs as point processes.
Section III discusses the statistical characteristics and distributions of 
signals and interference in the HCCN and approximates the desired signal with 
moment matching.
The coverage probability of the HCCN is obtained in Section IV by deriving 
the higher-order derivatives of the Laplace transform of the interference.
Section V analyzes the coupling between the aggregated desired signal and 
the interference and performs the Laplace transform to derive the achievable rate.
The simulation results of the performance analysis of the HCCN as well as the 
Monte-Carlo (MC)  validation are shown in Section VI.
Finally,   Section VII gives the conclusions.

\emph{Notations:}  
Boldface lower- and upper-case letters are used to denote
vectors and matrices,
respectively. 
The space of an $M\times N$-dimensional complex-valued matrix is denoted as $\mathbb{C}^{N\times M}$.
The 
circularly symmetric complex  Gaussian (CSCG)
distribution with mean $m$ and variance $\delta^{2}$ is denoted as $\mathcal{CN}(m,\delta^{2})$.
$\Gamma(k,\theta)$ denotes the Gamma distribution with shape parameter $k$ and scale parameter
$\theta$. $\mathrm{Nakagami}(m,\omega)$ 
denotes the Nakagami distribution with shape parameter $m$ and spread parameter $\omega$.
The Laplace transform of $Y$ is denoted as $  \mathcal{L}_{Y}(s) $.
$\left\lceil  \cdot \right\rceil$ and $\left\lfloor \cdot \right\rfloor$ denote the 
ceiling and flooring functions, respectively.
$\gamma(s,x)$ and $\Gamma(s,x)$ denote the upper and lower  incomplete Gamma function, respectively. 
In addition, $\Gamma(s)$ denotes the Gamma function. 

%  $\mathbf I_N$ denotes an $N \times N$ identity matrix.
% For a matrix $\mathbf{A}$, its
% transpose, conjugate, Hermitian transpose, and determinant  are respectively
% denoted as
% $\mathbf{A}^{\mathrm{T}}$,$\mathbf{A}^{*}$, $\mathbf{A}^{\mathrm{H}}$ and $| \mathbf{A}|$. 
% For a vector $\mathbf{a}$, $[\mathbf{a}]_{i}$ denotes its $i$-th element.

% $\mathcal{CN}(\mu,\sigma^{2})$  denotes

\section{System Model}
As shown in Fig. \ref{fig:CCCN},  an HCCN is considered
where the locations of BSs, cell-free APs, and
single-antenna  UEs  
are modeled by independent HPPP $\Lambda_{B}$, 
$\Lambda_{A}$ and $\Lambda_{U}$, 
with densities   denoted as $\lambda_{B}$/$\mathrm{km}^2$, 
$\lambda_{A}$/$\mathrm{km}^2$ and $\lambda_{U}$/$\mathrm{km}^2$, respectively.
The number of antennas for BS is $N_B$, while the number of antennas for each AP  is
 $N_A$.
Considering the difference in hardware specifications between BSs and APs, 
the maximum transmit power of BSs and APs is  $P_B$ and $P_A$  with  $P_B > P_A$.  
The entire network area is represented by $\mathcal{A}$.

 \begin{figure}[!t]
  \centering
    {\includegraphics[width=0.6\columnwidth]{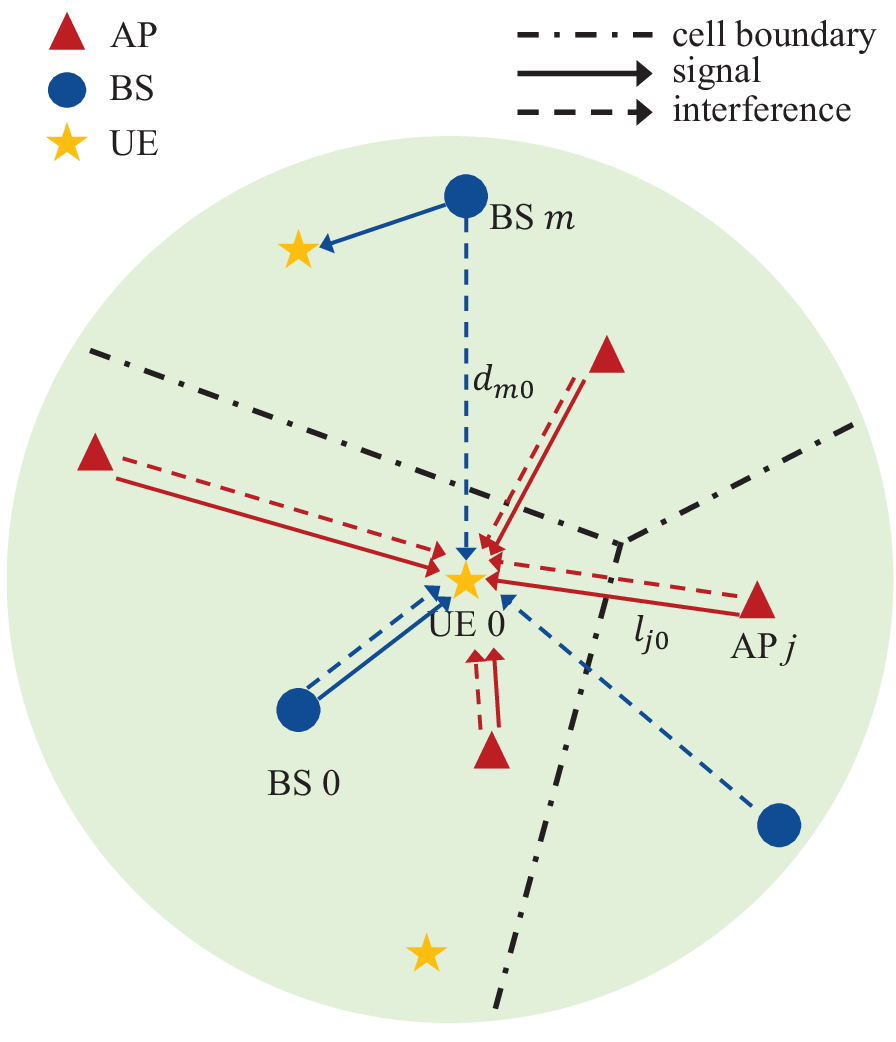}}%
  \caption{Structure of HCCNs,
  where  APs serve all UEs in the network simultaneously, while each BS serves 
  only UEs in the cellular area.}
  \label{fig:CCCN}
\end{figure}

In the HCCN, $d_{mi}$ and $l_{ji}$ denote the    BS $m$-UE $i$ and 
AP $j$-UE $i$ distances, respectively. 
We consider a typical UE, referred as UE 0, which is jointly served by the closest BS, 
named BS 0, with the distance $d_{00}$, and all the cell-free APs in $\mathcal{A}$,  denoted by $\omega_{A}$.
Besides, the set of all the  BSs in the network that contribute to the signal or cause interference to UE 0
is denoted as  $\omega_B$.
Based on the characteristics of cell-free networks, the HCCN is considered to operate in  time-division duplex  (TDD) mode, 
where both  BSs and  APs operate on the same frequency,
and are compatible with the existing cellular protocols
%instead of blocks of mutually orthogonal   time-frequency resource 
\cite{lijiamin2018downlink, Papazafeiropoulos2020performance,seifi2016adaptive}.
The channel vector between BS $m$ and 
UE $i$  is denoted by $\mathbf{h}_{mi} \triangleq [h_{mi,1},...,h_{mi,N_B} ] \in 
\mathbb{C}^{N_B\times 1 }$, while the channel vector between AP $j$ and UE $i$ is 
$\mathbf{g}_{ji} \triangleq [g_{ji,1},...,g_{ji,N_A} ]  \in \mathbb{C}^{N_A\times 1 }$. 
The channel model which consists of distance-dependent large-scale fading and 
random small-scale fading, can be formulated as
\begin{equation}\label{eq:hmi}
  \mathbf{h}_{mi}=\beta_{mi}^{\frac{1}{2}}\bm{\zeta}_{mi},
\end{equation} 
\begin{equation}\label{eq:gji}
  \mathbf{g}_{ji}=\delta_{ji}^{\frac{1}{2}}\bm{\xi}_{ji},
\end{equation}
where $\beta_{mi}$ and $\delta_{ji}$  are path loss 
%of the channel 
with  
$\beta_{mi}=\beta_0 d_{mi}^{-\alpha_1}$ and $\delta_{ji}=\delta_0 l_{ji}^{-\alpha_2}$.
Due to  the difference between AP and BS, 
different path-loss exponents $\alpha_1$ and $\alpha_2$ are used
in channel modeling.
The small-scale fading   in both  $\bm{\zeta}_{mi}$ and $\bm{\xi}_{ji}$ are 
independent and identically distributed (i.i.d.) $\mathcal{CN}(0,1)$  random variables.

\subsection{Conjugate Beamforming}

In HCCNs, each BS   provides service to all UEs in its
 cellular area.
%Thus the UE set $\phi_{B,m}$ served by BS $m$ at the current time slot has $|\phi_{B,m}|=1$.
To avoid the exchange of  channel state information (CSI) between   APs
and to obtain good performance with low  complexity \cite{ngo2017cell},
both BSs and APs transmit the signal via conjugate beamforming.
Thus, the signal sent by BS $m$  to the UEs it serves is 
\begin{equation}
  \mathbf{x}_{B,m}=\sqrt{P_B \eta_B} \sum_{n \in \phi_{B,m}}
  \frac{\mathbf{h}_{mn}}{\|\mathbf{h}_{mn} \|}q_{n},
\end{equation}
where $\phi_{B,m}$ denotes the   UE set
that are served by BS $m$, 
and $q_{n} \sim \mathcal{CN}(0,1)$ denotes the  information-bearing   symbol transmitted to 
UE $n$.  $\eta_B$ denotes the power constraint parameter with
$\mathbb{E}[  \mathbf{x}_{B,m}^H \mathbf{x}_{B,m} ]=P_B$.
Without of generality, $\eta_B$ is expressed as the mean UE number per BS, i.e.,
$\eta_B=\frac{1}{|\bar{\phi}_{B}|}=\frac{\lambda_B}{\lambda_U}$,
where $|\bar{\phi}_B|$ denotes the average of $|\phi_{B,m}|$ for any $m$.

Denote the set of UEs in HCCN area as $\phi_U$. Therefore, 
the corresponding downlink signal from the cell-free AP $j$   is 
\begin{equation}
  \mathbf{x}_{A,j}=\sqrt{P_A \eta_A} \sum_{i\in \phi_{U}} \frac{ \mathbf{g}_{ji}}{\|\mathbf{g}_{ji}\|}q_i,
\end{equation}
where $\eta_A$ denotes the power constraint parameter
with $\eta_A=\frac{1}{\bar{U}}$
to ensure $\mathbb{E}[  \mathbf{x}_{A,j}^H \mathbf{x}_{A,j} ]=P_A$.
Note that the average number of  UEs in $\phi_{U}$ is  
$\mathbb{E}[| \phi_{U}|]=\bar{U}=\lambda_U|\mathcal{A}|$.

For any UE $i$ in the network,  the 
index of  the associated BS with the closest distance that provides service
is denoted by $i^{*}$.
The aggregated downlink signal received by the typical UE 0 from the 
coordinated cell-free APs and  cellular BSs is 
\begin{equation}\label{eq:y0}
  \resizebox{1\hsize}{!}
  {$\begin{aligned}
   y_0&=\sum_{m\in \omega_B}\mathbf{h}_{m0}^H \mathbf{x}_{B,m} + 
  \sum_{j\in \omega_{A}}\mathbf{g}_{j0}^H\mathbf{x}_{A,j}+n_0
\\ &=\underbrace{\sqrt{P_B \eta_B}\|\mathbf{h}_{00} \|}_{S_{0B}}q_0
+\underbrace{\sum_{j\in \omega_{A}} \sqrt{P_A \eta_A}   \|\mathbf{g}_{j0}\|}_{S_{0A}}q_0+
\\ &  \underbrace{\sum_{i\in \phi_{U}\backslash_{0}}\! \Big(
  \sqrt{P_B\eta_B}\mathbf{h}_{i^{*}0}^H \frac{\mathbf{h}_{i^{*}i}}
  {\|\mathbf{h}_{i^{*}i} \|}
  \!+\!\sqrt{P_A \eta_A}  \sum_{j\in \omega_{A} }
  \mathbf{g}_{j0}^H \frac{ \mathbf{g}_{ji}}{\|\mathbf{g}_{ji}\|} \Big) }_{I_{U}}q_{i}+n_0,
% \\ &
% +\underbrace{\sqrt{P_B\eta_B}\mathbf{h}_{00}^H \sum_{n\in \phi_{B,0}\backslash \{0\}} 
% \frac{\mathbf{h}_{0n}}{\|\mathbf{h}_{0n} \|}q_n}_{C}
% \\ &+\underbrace{\sqrt{P_B\eta_B}\sum_{m\in \omega_B \backslash \{0\}} 
% \mathbf{h}_{m0}^H  \sum_{n\in \phi_{B,m}}
%  \frac{\mathbf{h}_{mn}}{\|\mathbf{h}_{mn} \|}q_{n}}_{D}
%  \\ &+ \underbrace{\sqrt{P_A \eta_A} \sum_{i\in \phi_{U}\backslash\{0\}} \sum_{j\in \omega_{A}}
%  \mathbf{g}_{j0}^H \frac{ \mathbf{g}_{ji}}{\|\mathbf{g}_{ji}\|}q_i}_{E}+n_0,
  \end{aligned} $}
\end{equation}
where   $S_{0B}   $   and    $S_{0A}$  in (\ref{eq:y0})
represent the desired signals from 
BS 0 and    cell-free APs, respectively. 
The third term $I_{U}$ denotes the total interference,
where each term includes the signals sent by the corresponding associated BS   
and cell-free APs to
any other UE. The last term $n_0 $ denotes  the additive white Gaussian noise (AWGN)
with zero mean and power $\sigma^2$. 

\subsection{Performance Metrics}

The coverage probability and   average achievable rate are both critical network
performance metrics. To derive the network coverage and average rate of HCCNs, the components of SINR should be analyzed first.
Based on  (\ref{eq:y0}), the  interference power caused   by the signal intended to 
UE $i$ is given by
\begin{equation}\label{eq:Ii}
  I_{i}=\Big|
    \sqrt{P_B\eta_B}\mathbf{h}_{i^{*}0}^H \frac{\mathbf{h}_{i^{*}i}}
    {\|\mathbf{h}_{i^{*}i} \|}
    \!+\!\sqrt{P_A \eta_A}  \sum_{j\in \omega_{A} }
    \mathbf{g}_{j0}^H \frac{ \mathbf{g}_{ji}}{\|\mathbf{g}_{ji}\|} \Big|^2.
\end{equation}

In interference $I_{i}$, not only AP channel vectors $\mathbf{g}_{i0}$ and $\mathbf{g}_{j0}$
 from different APs are independent of each other, $\forall i \neq j $, but also the BS channel vector 
$\mathbf{h}_{i^{*}0}$ is independent of    $\mathbf{g}_{j0}, \forall j \in \omega_{A}$.
According to the law of large numbers   and taking into account the mutual independence 
of channels as well as beamforming vectors between different BSs and APs, $I_{i}$ can be approximated as
\begin{equation}\label{eq:Iiapp}
  I_{i}\approx P_B\eta_B \Big|\mathbf{h}_{i^{*}0}^H \frac{\mathbf{h}_{i^{*}i}}
  {\|\mathbf{h}_{i^{*}i} \|}\Big|^{2}+P_A \eta_A\sum_{j\in \omega_{A}}
  \Big| \mathbf{g}_{j0}^H \frac{ \mathbf{g}_{ji}}{\|\mathbf{g}_{ji}\|}\Big|^{2}.
\end{equation}

% According to the approximation   in (\ref{eq:Iiapp}), 
% $I_{i}$ can be expressed as the sum of the 
%  signal powers from each AP and   the associated BS $i^*$. 
Further,
the total interference can be recast as follows by classifying it into  
intra-cell interference $I_{B0}$, 
inter-cell interference $I_{B}$,
and  interference $I_{A}$ produced by AP signals to other UEs
\begin{equation}
  \sum_{i\in \phi_{U}\backslash\{0\}}I_{i}=I_{B0}+I_{B}+I_{A},
\end{equation}
where 
\begin{equation}\label{eq:SI}
  \begin{aligned}
    I_{B0}&= P_B\eta_B \sum_{n\in \phi_{B,0}\backslash \{0\}} 
    \Big| \mathbf{h}_{00}^H \frac{\mathbf{h}_{0n}}{\|\mathbf{h}_{0n} \|}\Big|^2,\\
    I_{B}&= P_B \eta_B\sum_{m\in \omega_B \backslash \{0\}} 
    \sum_{n\in \phi_{B,m}}
    \Big| \mathbf{h}_{m0}^H \frac{\mathbf{h}_{mn}}{\|\mathbf{h}_{mn} \|} \Big|^2,\\
    I_{A}&=P_A \eta_A  \sum_{i\in \phi_{U}\backslash\{0\}}\sum_{j\in \omega_{A}} \Big|   
    \mathbf{g}_{j0}^H \frac{ \mathbf{g}_{ji}}{\|\mathbf{g}_{ji}\|}  \Big|^2.
  \end{aligned}
\end{equation}

% \begin{figure}[!t]
%   \centering
%     {\includegraphics[width=0.72\columnwidth]{./figure/appofI.eps}}%
%   \caption{ Approximation of (7) to (6).}
%   \label{fig:appofI}
% \end{figure}

With the analysis of the components of   the interference,
the corresponding received SINR 
is approximately  expressed as 
\begin{equation}\label{eq:SINR}
  \Omega=\frac{S_0}{I_{B0}+I_{B}+I_{A}+\sigma^2},
\end{equation}
where 
\begin{equation}
  \resizebox{1\hsize}{!}
   {$
  \begin{aligned}
  S_0&\!=\!\big( S_{0B}\!+\!S_{0A}\big)^2 \!=\! \Big( \sqrt{P_B \eta_B}\|\mathbf{h}_{00} \|   
  \!+\!\sqrt{P_A \eta_A}  \sum_{j\in \omega_{A}}   \|\mathbf{g}_{j0}\| \Big)^2.
  \end{aligned}
 $}
\end{equation}

% \Big| \sqrt{P_B\eta_B}\frac{\mathbf{h}_{00}^H\mathbf{h}_{00}}{\|\mathbf{h}_{00} \|}   
%     +\sum_{j\in \omega_{A,0}} \sqrt{P_A \eta_A}    \frac{\mathbf{g}_{j0}^H \mathbf{g}_{j0}}{\|\mathbf{g}_{j0}\|}  \Big|^2
%     \\ &=

Therefore, the   coverage probability is written    as 
\begin{equation}
  p_{\mathrm{c}}\triangleq\mathbb{P}[ \Omega > T]=P[\frac{S_0}{I_{B0}+I_{B}+\bar{I}_{A}+\sigma^2}> T],
\end{equation}
while the average achievable rate is 
\begin{equation}
  R =\mathbb{E}\big[\mathrm{ln}(1+\Omega) \big]
  =\mathbb{E}\big[\mathrm{ln}(1+\frac{S_0}{I_{B0}+I_{B}+\bar{I}_{A}+\sigma^2}) \big],
\end{equation} 
where the average is taken over both the random spatial 
distribution and  the channel fading  of UEs, BSs and APs.

\section{Analysis of Signal Strength and Interference}
The statistical characteristics of signal and interference power 
are first characterized.
The  power distributions corresponding to the 
isotropic channel with   i.i.d. entries have been derived in \cite{Chandrasekhar2009coverage}.
After deriving the closed forms of the first- and second-order moments, 
the signal strength is approximated by moment matching.
The intra- and inter-cell interference caused by BSs 
is approximated as  the sum of one or more Gamma random variables. 
The average interference caused by cell-free APs is derived in closed-form.

% In this section, 
% the statistical distributions of the received signal and interference power 
% are characterized.
% It is well known that the signal and interference power of each UE 
% is proportional to the power  of the  channel projected onto
% the beamforming subspace. 
% The signal and interference power distributions corresponding to the 
% isotropic channel with independent and identically distributed (i.i.d.) entries have been derived \cite{Chandrasekhar2009coverage}.
% However,   the channels of each UE in network are simultaneously affected 
% by multiple path-loss parameters from  BSs as well as cell-free APs, 
% and the analysis of such non-isotropic channels is more challenging.
% This section first  provides the analysis of the distribution of channel strength 
% in the hybrid cellular and cell-free network with respect to signals and interference.
% Then, the signal strength is approximated by a r.v. with Gamma distribution, 
% with its first-   and second-order moments derived in closed-form. 
% Besides, the intra- and inter-cell interference caused by BSs 
% are approximated a Gamma r.v. and a weighted sum of Gamma  r.v., 
% respectively. The average interference caused by cell-free APs is derived in closed-form.
% In addition, the performance of the network 
% can be obtained by analyzing a typical UE 0 according to Slivnyak's theorem \cite{andrews2011tractable}.
\subsection{  Channel Distribution}

% 
% The analysis of the power of the BS  channel $\mathbf{h}_{mi}$ and the AP channel 
% $\mathbf{g}_{ji}$ is the basis for the  derivation of the signal and interference power distribution.
Based on  channel model of BSs and APs,
the channel powers  from the $m$th BS and the $j$th AP to UE $i$ are
\begin{equation}\label{eq:hmi2}
  |\mathbf{h}_{mi}|^2=\beta_{mi}\bm{\zeta}_{mi}^{H}\bm{\zeta}_{mi},
\end{equation}
\begin{equation}\label{eq:gji2}
  |\mathbf{g}_{ji}|^2=\delta_{ji}\bm{\xi}_{ji}^{H}\bm{\xi}_{ji}.
\end{equation}

Since all the entries  
follow the i.i.d. $\mathcal{CN}(0,1)$,
$\bm{\zeta}_{mi}$ and $\bm{\xi}_{ji}$  are isotropic vectors of $N_B$ and $N_A$
dimensions respectively \cite{  vershynin2018high }.
Note that for the isotropic vector $\mathbf{x}\in \mathbb{C}^{N\times 1}$ with each entry
following i.i.d. $\mathcal{CN}(1,\delta^2)$, $\mathbf{x}^H\mathbf{x}$ is the sum of i.i.d.
variables $\Gamma(1,\delta^2)$, and thus follows $\Gamma(N,\delta^2)$ \cite{health2011multiuser}. 
Therefore, we have $\bm{\zeta}_{00}^{H}\bm{\zeta}_{00}\sim \Gamma(N_B,1) $ and 
$\bm{\xi}_{j0}^{H}\bm{\xi}_{j0}\sim \Gamma(N_A,1)$.
% Further, Lemma \ref{lem1} is introduced as follows.
\begin{lemma}
  \label{lem1}
For  the Gamma distributed random variable $X \sim \Gamma(a,\theta)$ and any $b>0$,
  $Y=bX\sim \Gamma(a,b\theta)$ \cite{moschopoulos1985distribution}.
\end{lemma}

Based on  Lemma \ref{lem1}, the BS and AP channel powers  in 
(\ref{eq:hmi2}) and (\ref{eq:gji2}) have the following distributions
\begin{equation}\label{eq:h2sim}
  |\mathbf{h}_{mi}|^2\sim \Gamma(N_B,\beta_{mi}),
\end{equation}
\begin{equation}
  |\mathbf{g}_{ji}|^2\sim \Gamma(N_A,\delta_{ji}).
\end{equation}

% Note that the distributions of $|\mathbf{h}_{mi}|^2$ and $|\mathbf{g}_{ji}|^2$
% apply to any BS-to-UE pair, as well as AP-to-UE pairs.
% and thus can be utilized in the approximation of the 
% power distributions of both signal and interference.

\subsection{  Signal Power Distribution}
According to (\ref{eq:h2sim}),
 the power of the nearest associated BS channel $|\mathbf{h}_{00}|^2$
is the sum of $N_B$ i.i.d. variables
following $\Gamma(1,\beta_{00})$, i.e., $|\mathbf{h}_{00}|^2\sim 
\Gamma(N_B,\beta_{00})$. 
To further analyze  the desired signal $S_0$ in (\ref{eq:SI}), the following 
Lemma \ref{lem5}   is  introduced.
\begin{lemma}
  \label{lem5}
For  any Gamma  random variable $X \sim \Gamma(k,\theta)$,
its square root $Y$   follows the Nakagami distribution
as $Y=\sqrt{X}\sim \mathrm{Nakagami}(m,\omega)$ 
with $m=k, \omega=m\theta$ \cite{huang2016nakagami}.
\end{lemma}

Based on the  Lemma \ref{lem5}, the distribution of $\|\mathbf{h}_{00} \| $ is  
\begin{equation}
  \|\mathbf{h}_{00} \|=\sqrt{|\mathbf{h}_{00}|^2}\sim \mathrm{Nakagami}(N_B, N_B \beta_{00}),
\end{equation}
while the component  $ \|\mathbf{g}_{j0} \|$ about AP channel in $S_0$ has
\begin{equation}
  \|\mathbf{g}_{j0} \|=\sqrt{|\mathbf{g}_{j0}|^2}\sim \mathrm{Nakagami}(N_A, N_A \delta_{j0}), \forall j\in \omega_A. 
\end{equation}

The distribution of  the desired signal $S_0$ consists of the signals 
from the associated BS 0 together with all APs.
Considering the high-density random distribution of the cell-free  APs 
in the network,   Lemma \ref{lemassLA}
is introduced.

\begin{lemma}
  \label{lemassLA}
  With  the law of large numbers and the Campbell Theorem,  
  the  signal   due to the large number of  APs  in $S_0$  can be approximated 
  by their average $L_A$ when   $\alpha_2<4$, i.e., 
\begin{equation}
    \begin{aligned}
    \sqrt{P_A \eta_A}  &\sum_{j\in \omega_{A}}  
  \|\mathbf{g}_{j0}\|  \approx \sqrt{P_A \eta_A}   \mathbb{E}[\sum_{j\in \omega_{A}}  
  \|\mathbf{g}_{j0}\|] \\
   &= \underbrace{\frac{4\pi\sqrt{\rho_A} \lambda_A \delta_{0}^{\frac{1}{2}}}{4-\alpha_2}
   \frac{\Gamma(N_A+\frac{1}{2})}{\Gamma(N_A)}  
   \Big(\frac{|\mathcal{A}|}{\pi}\Big)^{1-\frac{\alpha_2}{4}}}_{L_A},
    \end{aligned}
  \end{equation}     
  where $\rho_A=P_A \eta_A$. 
\end{lemma}
\begin{proof}
  %The derivation is omitted  due to space limitation.
  % Please refer to Appendix \ref{Appendix:LA} for detailed derivation.

  The average $L_A$   is obtained by averaging both 
  the AP spatial distribution and
  the channel fading as 
\begin{equation}
  \begin{aligned}
    L_A&=\sqrt{\rho_A} \mathbb{E}_{\Lambda_A}\Big[ \sum_{j\in \omega_{A}}  
    \mathbb{E}_{\bm{\xi}}[ \|\mathbf{g}_{j0}\|]\Big]\\
    &\overset{(a)}{=} \sqrt{\rho_A} \mathbb{E}_{\Lambda_A}\Big[ \sum_{j\in \omega_{A}} 
    \frac{\Gamma(N_A+\frac{1}{2})}{\Gamma(N_A)} \delta_{j0}^{\frac{1}{2}} \Big]\\
    % &=\sqrt{\rho_A} \mathbb{E}_{\Lambda_A}\Big[ \sum_{j\in \omega_{A}}  
    % \frac{\Gamma(N_A+\frac{1}{2})}{\Gamma(N_A)} \delta_{0}^{\frac{1}{2}} 
    %  l_{j0}^{-\frac{\alpha_2}{2}} \Big]\\
     &\overset{(b)}{=}
     2\pi\sqrt{\rho_A} \lambda_A \delta_{0}^{\frac{1}{2}}
      \frac{\Gamma(N_A+\frac{1}{2})}{\Gamma(N_A)} \int_{0}^{\sqrt{\frac{|\mathcal{A}|}{\pi}}}
    l^{-\frac{\alpha_2}{2}+1}   \mathrm{d}l\\
    &=
    \frac{4\pi\sqrt{\rho_A} \lambda_A \delta_{0}^{\frac{1}{2}}}{4-\alpha_2}
    \frac{\Gamma(N_A+\frac{1}{2})}{\Gamma(N_A)}  
    \Big(\frac{|\mathcal{A}|}{\pi}\Big)^{1-\frac{\alpha_2}{4}},
  \end{aligned}
\end{equation}
where relation (a) is obtained since 
$ \|\mathbf{g}_{j0} \|\sim \mathrm{Nakagami}(N_A, N_A \delta_{j0}), 
\forall j\in \omega_A$, and the mean of Nakagami variable with $\mathrm{Nakagami}(m,\omega)$
is $\frac{\Gamma(m+\frac{1}{2})}{\Gamma(m)}\big(\frac{\omega}{m}\big)^{\frac{1}{2}}$. 
Relation (b) is obtained by applying the Campbell Theorem
that 
for the general PP $\Lambda$ with density $\lambda$ 
and a measurable function $f: \mathbb{R}^n \to \mathbb{R} $, 
the sum of $f$ over the PP is
\begin{equation}\label{eq:campbell}
  \mathbb{E}\Big[\sum_{x\in \Lambda} f(x)  \Big]
  = \lambda \int_{\Lambda}f(x)\mathrm{d}x.
\end{equation}
This completes the derivation of $L_A$.
\end{proof}

\begin{figure}[!t]
  \centering
    {\includegraphics[width=0.72\columnwidth]{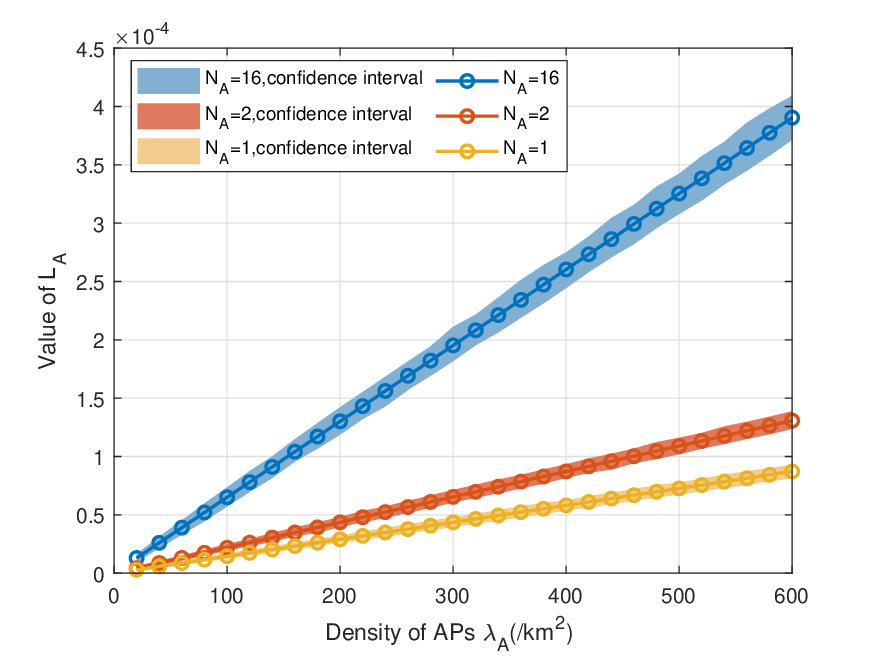}}%
  \caption{The average value and 95\% confidence interval for $L_A$.}
  \label{fig:LA-interval}
\end{figure}

% \begin{figure}[!t]
%   \centering
%     {\includegraphics[width=0.72\columnwidth]{./figure/LA-errorinterval.eps}}%
%   \caption{ Percentage of  error for $L_A$.}
%   \label{fig:LA}
% \end{figure}

Fig. \ref{fig:LA-interval} shows the  value and 95\% confidence interval for 
$L_A$ with different number of AP antenna $N_A$, 
where $\alpha_2=1.2$, $\lambda_U= 40/\mathrm{km}^2$.
Besides, Fig. \ref{fig:LA-interval} demonstrates the feasibility and effectiveness of 
approximating the mean of the desired signal due to the APs with $L_A$.

% \begin{assumption}
%   \label{assLA}
%   By means of the law of large number and the Campbell Theorem, 
% the desired signal in $S_0$   due to the APs can be approximated 
% by their average $L_A$ where the number of APs is large, i.e., 
% \begin{equation}
%   \sqrt{P_A \eta_A}  \sum_{j\in \omega_{A}}  
% \|\mathbf{g}_{j0}\|\approx \sqrt{P_A \eta_A}  \sum_{j\in \omega_{A}}  
%  \mathbb{E}[\|\mathbf{g}_{j0}\|] .
% \end{equation} 

% \end{assumption}
% Hence, the desired signal power contributed by the APs can be considered as a constant $L_A$, 
% which is independent of the signal power contributed by the BS. 
% The expression of $L_A$ is 
% \begin{equation}
%   L_A= \frac{4\pi\sqrt{\rho_A} \lambda_A \delta_{0}^{\frac{1}{2}}}{4-\alpha_2}
%   \frac{\Gamma(N_A+\frac{1}{2})}{\Gamma(N_A)}  
%   \Big(\frac{|\mathcal{A}|}{\pi}\Big)^{1-\frac{\alpha_2}{4}},
% \end{equation}

% \begin{figure}[!t]
%   \centering
%     {\includegraphics[width=0.7 \columnwidth]{./figure/LA.eps}}%
%   \caption{Verification of $L_A$ in 40 Monte-Carlo (MC) generated maps, 
%   where random channels are generated 5 times for each map.}
%   \label{fig:LA}
% \end{figure}

% The verification of Lemma \ref{lemassLA} is shown in Fig. \ref{fig:LA}.
From Lemma \ref{lemassLA}, the power expression  $S_0$ for the desired signal 
is reduced to the square of the sum of  a Nakagami random variable and the constant $L_A$
as $S_0\approx ( \sqrt{P_B \eta_B}\|\mathbf{h}_{00} \|   +L_A )^2$.
Therefore, the following Lemma is introduced.

\begin{lemma}
  \label{lemassshifted}
For any Nakagami random variable $X\sim  \mathrm{Nakagami}(m,\omega)$, the 
probability density function (PDF) of the 
square of  shifted  Nakagami random variable $Y=(X+A)^2$ for $Y>A^2$ is 
\begin{equation}\label{eq:fYy}
  \resizebox{1\hsize}{!}
  {$
  \begin{aligned}
  f_{Y}(y)
  % &\!=\!\frac{2m^m}{\Gamma(m)\omega^m}\!(\!\sqrt{y}\!-\!A)^{2m\!-\!1}
  % \mathrm{exp}\!\big(\!-\!\frac{m}{\omega}(\sqrt{y}\!-\!A)^2\big)
  % \!\Big|\frac{\mathrm{d}(\sqrt{y}\!-\!A)}{\mathrm{d}y} \Big|\\
   & =\!\frac{m^m}{\Gamma(m)\omega^m}\!(\!\sqrt{y}\!-\!A)^{2m\!-\!1}
   \mathrm{exp}\!\big(\!-\!\frac{m}{\omega}(\sqrt{y}\!-\!A)^2\big)y^{\frac{1}{2}}.
\end{aligned}$}
\end{equation}
\end{lemma}

% \begin{figure}[!t]
%   \centering
%     {\includegraphics[width=0.72\columnwidth]{./figure/X+A2.eps}}%
%   \caption{The PDF of the 
%   square of  shifted  Nakagami random variable.}
%   \label{fig:X+A2}
% \end{figure}

% Fig. \ref{fig:X+A2} displays the validation of the distribution of $Y=(X+A)^2$.
From (\ref{eq:fYy}), it is challenging to characterize
the exact distribution of $S_0$. However, the 
 PDF of $S_0$ has a similar structure to  the Gamma distribution.
Therefore, given the distance to the   associated BS 0, 
$S_0$ is approximated as a Gamma random variable based on its
first- and second-order moments \cite{Lyu2021IRS}. 
The corresponding Lemma \ref{lemappS0} is introduced as follows.
\begin{lemma}
  \label{lemappS0}
  According to the definition of the Gamma random variable \cite{pishro2014introduction}, the desired signal power
  $S_0$ can be approximated as  the Gamma distribution $\Gamma(k_{S_0},\theta_{S_0})$ with
  \begin{equation}\label{eq:apppara}
    \begin{aligned}
      k_{S_0}&=\frac{\big( \mathbb{E}[S_0] \big)^2}{ \mathrm{Var}\{S_0\} }
      =\frac{\big( \mathbb{E}[S_0] \big)^2}{ \mathbb{E}[S_0^2] -\big( \mathbb{E}[S_0] \big)^2},\\
      \theta_{S_0}&=\frac{ \mathrm{Var}\{S_0\} }{\mathbb{E}[S_0] }
      =\frac{ \mathbb{E}[S_0^2] -\big( \mathbb{E}[S_0] \big)^2}{\mathbb{E}[S_0] },
    \end{aligned}
  \end{equation}
where  the first- and second-order moments of $S_0$ are 
\begin{equation}\label{eq:ES_0}
  \mathbb{E}[S_0]=\rho_B  \beta_{00} N_B
  +2\rho_B^{\frac{1}{2}} \beta_{00}^{\frac{1}{2}} L_A\frac{\Gamma(N_B +\frac{1}{2})}{\Gamma(N_B)}
  +L_{A}^{2},
\end{equation}
\begin{equation}\label{eq:ES_02}
  \begin{aligned}
  \mathbb{E}[S_0^2] &=\rho_B^2\beta_{00}^2 N_B(N_B+1)
  +4\rho_B^{\frac{3}{2}}\beta_{00}^{\frac{3}{2}}L_A \frac{\Gamma(N_B +\frac{3}{2})}{\Gamma(N_B)}
  \\ & +6\rho_B\beta_{00} L_A^2 N_B 
  +4\rho_B^{\frac{1}{2}}\beta_{00}^{\frac{1}{2}}L_A^3\frac{\Gamma(N_B +\frac{1}{2})}{\Gamma(N_B)}+L_{A}^{4} .
\end{aligned}
\end{equation}

\end{lemma}
\begin{proof}
  %The derivation is omitted  due to space limitation.
  % The derivation is mainly based on the raw moments of  Gamma distribution 
  % and the mean of Nakagami distribution. For brevity, the detailed derivation 
  % will be shown in the journal version.
 Please refer to Appendix \ref{Appendix:moment}.
\end{proof}

\subsection{Interference Power Distribution}

% This subsection analyzes the interference power distribution in an approach similar 
% to that of the desired signal. 
Considering the conjugate beamforming design,
Lemma \ref{lem3} is introduced. 
\begin{lemma}
 \label{lem3}
Denote $\mathbf{x} \in \mathbb{C}^{N\times1}$ as an isotropic vector with i.i.d.
$\mathcal{CN}(0,\theta)$ entries.
If $\mathbf{x}$ is projected onto an s-dimensional beamforming subspace, the  
power distribution is \cite{hosseini2016stochastic,muirhead2009aspects}
\begin{equation}
 |\mathbf{x}^H\mathbf{w}|^2\sim \Gamma(s,\theta).
\end{equation} 
\end{lemma}

% Denote the approximation of the non-isotropic composite vector $\mathbf{f}$ in network 
% by an isotropic vector $\mathbf{f}_a$ with i.i.d. $\mathcal{CN}(0,\theta_a)$. 
% Further, from each UE's perspective, the  intended beam lies in a  
% $(N_B+|\omega_{A}|N_A)$-dimensional subspace with the normalized conjugate beamforming \cite{lijiamin2018downlink}.
% When projected onto the conjugate beamforming $\mathbf{w}_c$, the isotropic approximation $\mathbf{f}_a$ has 
% \begin{equation} \label{eq:fbarw}
%  |\mathbf{f}_a^H\mathbf{w}_c|^2\sim \Gamma(N_B+|\omega_{A}|N_A,\theta_a).
% \end{equation}

As a consequence of Lemma \ref{lem3}, the
power $I_{B0}$ of intra-cell interference in (\ref{eq:SI}) is approximated as the sum of 
$(|\bar{\phi}_{B}|-1)$ i.i.d. variables following $\Gamma(1,P_B\eta_B\beta_{00})$.
Therefore, the approximation distribution of $I_{B0}$ with Lemma \ref{lem3} is 
\begin{equation}\label{eq:IB0-1}
  I_{B0}\sim \Gamma(|\bar{\phi}_{B}|-1,P_B\eta_B \beta_{00}).
\end{equation}

Further, extracting the scale parameter,
the power of intra-cell interference can be rewritten as 
$I_{B0}=P_B\eta_B \beta_{00}\kappa_{B,0}$, where $\kappa_{B,0}\sim 
\Gamma( |\bar{\phi}_{B}|-1,1) $.

For the inter-cell interference, since $\mathbf{h}_{m0}$ is also independent of 
$\mathbf{h}_{mn}$, the interference power of each BS $m\in \omega_B \backslash \{0\}$ 
in $I_{B}$ is approximated as
$\sum_{n\in \phi_{B,m}}
\Big| \mathbf{h}_{m0}^H \frac{\mathbf{h}_{mn}}{\|\mathbf{h}_{mn} \|} \Big|^2
\sim \Gamma( |\bar{\phi}_{B}|,\beta_{m0})$.
Therefore,  
the inter-cell interference $I_{B}$ can be further rewritten as the sum of Gamma variables 
with the same shape parameters  and scale parameters, i.e.,
\begin{equation}\label{eq:IBkappa}
  I_{B}=P_B\eta_B \sum_{m\in\omega_B \backslash \{0\}} \beta_{m0}\kappa_{B,m0},
\end{equation}
where $\kappa_{B,m0}\in \Gamma( |\bar{\phi}_{B}|,1), \forall m \in\omega_B \backslash \{0\}$.

Next, for the analysis of the interference  $I_A$ from the APs, 
the following Lemma  is introduced based on (\ref{lemassLA}) and  (\ref{lem3}).
% and the 
% verification of Lemma \ref{lembarIA} is shown in Fig. \ref{fig:barIA}.
\begin{lemma}\label{lembarIA}
  By   the law of large numbers and the Campbell Theorem,  
  the interference $I_A$   due to the large number of  APs is approximated 
  by its average $\bar{I}_{A}$ when  $\alpha_2<2$, i.e., 
\begin{equation}
  \begin{aligned}
  I_{A}
   &=P_A \eta_A  \sum_{i\in \phi_{U}\backslash\{0\}}\sum_{j\in \omega_{A}} \Big|   
     \mathbf{g}_{j0}^H \frac{ \mathbf{g}_{ji}}{\|\mathbf{g}_{ji}\|}  \Big|^2\\
  &\approx P_A \eta_A \mathbb{E} \bigg[\sum_{i\in \phi_{U}\backslash\{0\}}\sum_{j\in \omega_{A}} \Big|   
    \mathbf{g}_{j0}^H \frac{ \mathbf{g}_{ji}}{\|\mathbf{g}_{ji}\|}  \Big|^2 \bigg]\\
    &= \underbrace{\frac{2\pi \rho_A  \lambda_A \delta_{0}(\lambda_U|\mathcal{A}|-1)}
    {2-\alpha_2}
    \Big(\frac{|\mathcal{A}|}{\pi}\Big)^{1-\frac{\alpha_2}{2}}}_{\bar{I}_{A}}.
  \end{aligned}
  \end{equation} 
\end{lemma}
\begin{proof}
  % Please refer to Appendix \ref{Appendix:bar_IA} for detailed derivation.

  The average  $\bar{I}_{A}$  
  is obtained by averaging both the AP spatial distribution and the channel
fading as  
\begin{equation}\label{eq:derivebarIA}
  \begin{aligned}
    \bar{I}_{A}&=\rho_A \mathbb{E}_{\Lambda_A} \bigg[ \sum_{j\in \omega_{A}}
        \mathbb{E}_{   \Lambda_U}\Big[
          \sum_{i\in \phi_{U}\backslash\{0\}}\big|   
    \mathbf{g}_{j0}^H \frac{ \mathbf{g}_{ji}}{\|\mathbf{g}_{ji}\|}  \big|^2 \Big] \bigg] \\
    &\overset{(a)}{=} \rho_A \mathbb{E}_{\Lambda_A} \bigg[ \sum_{j\in \omega_{A}}
     (\lambda_{U}|\mathcal{A}|-1)\delta_{j0}\bigg] \\
    &\overset{(b)}{=} \rho_A  (\lambda_{U}|\mathcal{A}|-1) \delta_0 
    \int_{0}^{\sqrt{\frac{|\mathcal{A}|}{\pi}}}2\pi \lambda_A l^{1- \alpha_2 } \mathrm{d}l\\
    &= \frac{2\pi\rho_A \lambda_A \delta_{0}(\lambda_U|\mathcal{A}|-1)}
    {2-\alpha_2} 
    \Big(\frac{|\mathcal{A}|}{\pi}\Big)^{1-\frac{\alpha_2}{2}}
   ,
  \end{aligned}
\end{equation}
where relation (a) comes from that 
$ \big|   
\mathbf{g}_{j0}^H \frac{ \mathbf{g}_{ji}}{\|\mathbf{g}_{ji}\|}  \big|^2
 \sim \Gamma(1,  \delta_{j0}), 
\forall i\in \phi_{U}\backslash\{0\}$. 
Relation (b) is obtained based on the  Campbell Theorem in (\ref{eq:campbell}).
Note that there is $\alpha_2<2$ in (\ref{eq:derivebarIA}).
This completes the derivation of $\bar{I}_{A}$.

\begin{figure}[!t]
  \centering
    {\includegraphics[width=0.72\columnwidth]{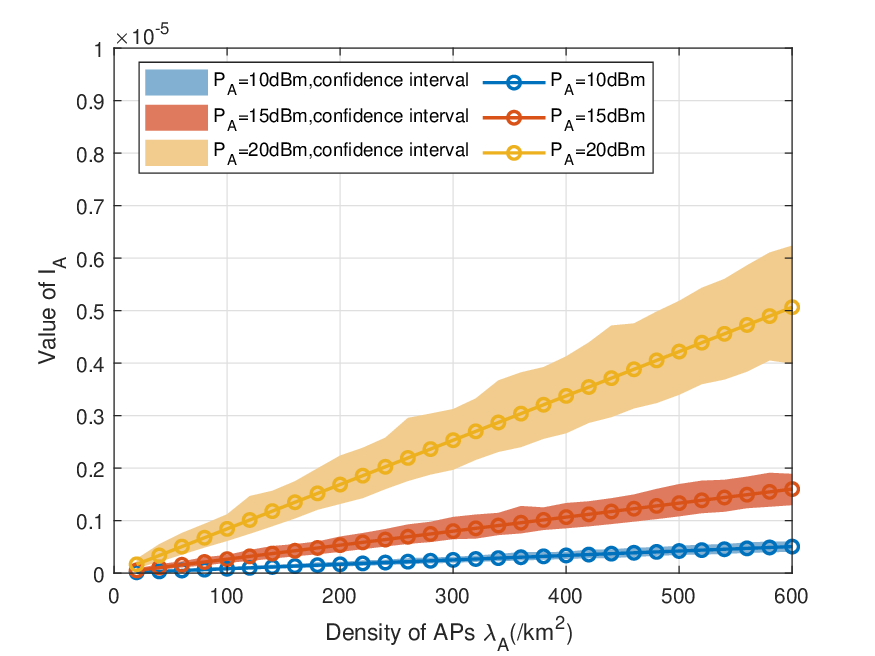}}%
  \caption{ The average value and 95\% confidence interval for $I_A$.}
  \label{fig:IA-interval}
\end{figure}

Fig. \ref{fig:IA-interval} shows the  value and 95\% confidence interval for 
$L_A$, 
where  $\lambda_U= 40/\mathrm{km}^2$.
From (\ref{eq:derivebarIA}), the value of $I_A$
does not change with $N_A$, and is hard to be affected by $\lambda_U$
since $\rho_A=\frac{P_A}{\lambda_U|\mathcal{A}|}$. Therefore, the value of $I_A$
with different $P_A$ is shown in Fig. \ref{fig:IA-interval}, which
demonstrates the feasibility and effectiveness of 
approximating the mean of the desired signal due to the APs with $I_A$.

\end{proof}

\section{Coverage Probability Analysis}
Based on the statistical characteristics of signal strength and various interference 
components derived in the preceding section, the  coverage probability  of HCCN is further analyzed.  
In general, the coverage probability is the  CCDF of SINR over the whole network, 
which can be defined as 
\begin{equation}\label{eq:covdef}
  p_{\mathrm{c}}\triangleq\mathbb{P}[ \Omega=\frac{S_0}{I_{B0}+I_{B}+\bar{I}_{A}+\sigma^2} > T],
\end{equation}
where $T$ denotes the target SINR threshold.

% \begin{equation}
%   \bar{I}_{A}=\frac{2\pi\delta_0 \rho_A}{2-\alpha_2}(|\bar{\phi}_{U}|-1)\frac{|\mathcal{A}|}{\pi}^{1-\frac{\alpha_2}{2}}
% \end{equation}
\subsection{  Coverage Probability with the  Nearest BS Distance} 
Taking the distance $d_{00}$ between   UE 0 and the nearest 
BS 0 as a random variable, the average  coverage probability in  the network is 
\begin{equation}\label{eq:cov1}
  \begin{aligned}
    p_{\mathrm{c}}&=\mathbb{E}\big[ p_{\mathrm{c}}(d_{00})  \big] = 
    \int_{0}^{\sqrt{\frac{|\mathcal{A}|}{\pi}}} p_{\mathrm{c}}(r) f_{d_{00}}(r) 
    \mathrm{d}r,
  \end{aligned}
\end{equation}
where $f_{d_{00}}(r)$ denotes the PDF of 
the  nearest BS distance in the entire network area.  With the
 cumulative distribution function (CDF) of  $d_{00}$ denoting as 
$F_{d_{00}}(r)=1-e^{-  \lambda_B \pi r^2}$ \cite{andrews2011tractable}, the
nearest point distance in PPP is
\begin{equation}
  \begin{aligned}
    f^{*}_{d_{00}}(r)&=\frac{\mathrm{d}F_{d_{00}}(r)}{\mathrm{d}r} =2\lambda_B \pi r e^{-\lambda_B \pi r^2}.
  \end{aligned}
\end{equation}

Considering that the radius of the network area is $\sqrt{\frac{|\mathcal{A}|}{\pi}}$, 
the   PDF of the nearest BS distance can be obtained after normalization as
\begin{equation}\label{eq:fd00}
  f_{d_{00}}(r)=\frac{f^{*}_{d_{00}}(r)}{p_{\mathrm{area}}},
\end{equation}
where 
$
  p_{\mathrm{area}}= \int_{0}^{\sqrt{\frac{|\mathcal{A}|}{\pi}}} f^{*}_{d_{00}}(r)
  \mathrm{d}r.
$

First, 
the coupling relationship between intra-cell interference $I_{B0}$ and inter-cell interference
$I_B$ is  analyzed.
It is evident that both $I_{B0}$ and $I_B$ rely on the distance $d_{00}$, i.e., 
$I_{B0}$ comes from the associated BS 0 with a distance of $d_{00}$, 
and $I_B$ comes from the other BSs with distances  $d\geq d_{00}$.
However, there is no interaction between  $I_{B0}$ and $I_B$.
Specifically, for any $d_{00}$, $I_{B0}$ depends  on the distribution of UEs
within the cell of BS 0, while $I_B$ depends mainly on the distribution of other BSs 
with a distance no less than $d_{00}$.
Thus, $I_{B0}$ and $I_B$ are independent of 
each other for a given $d_{00}$  \cite{Tanbourgi2014dual}.
Therefore,  Lemma \ref{lemcoverage} is introduced as follows.
\begin{lemma}
  \label{lemcoverage}
  Considering that the desired signal $S_0$ can be approximated as the Gamma distribution, 
  i.e.,  $S_0\sim\Gamma(k_{S_0},\theta_{S_0})$,
  the network coverage probability in (\ref{eq:covdef}) is expressed as
\begin{equation}\label{eq:outtolap}
  \begin{aligned}
  &p_{\mathrm{c}}(d_{00})=\mathbb{P}[ S_{0} > T(I_{B0}+I_{B}+\bar{I}_{A}+\sigma^2)]\\ 
  &=\sum_{i=0}^{k_{S_0}-1} \frac{(-1)^i}{i!} \frac{\partial^i}{\partial^i s}
  \Big\{ e^{-s\frac{TI_{e}}{\theta_{S_0}}}
      \mathcal{L}_{Y_{I_{B0}}|d_{00}}(s) 
      \mathcal{L}_{Y_{I_{B}}|d_{00}}(s)   \Big\}_{s=1},
  \end{aligned}
\end{equation}
where  $k_{S_{0}}$ is an integer. 
The sum of AP interference   and noise is denoted as $I_{e}=\bar{I}_{A}+\sigma^2$. The Laplace transforms 
of   
$Y_{I_{B0}}=\frac{TI_{B0}}{\theta_{S_0}}$ and $Y_{I_{B}}=\frac{TI_{B}}{\theta_{S_0}}$
are 
\begin{equation}\label{eq:YIlap}
  \resizebox{1\hsize}{!}
  {$\begin{aligned}
  &\mathcal{L}_{Y_{I_{B0}} |d_{00}}(s)=\Big(
    1+s\frac{T\rho_B\beta_0 d_{00}^{-\alpha_1}}{\theta_{S_0}}\Big)^{1-|\bar{\phi}_{B}|},\\
    &\mathcal{L}_{Y_{I_{B}}|d_{00} }(s) = 
     \mathrm{exp}\Big( 2\pi \lambda_B \int_{d_{00}}^{\sqrt{\frac{|\mathcal{A}|}{\pi}}} \big[
      \big(1+s\frac{T \rho_B \beta_{0} r^{-\alpha_1} }{\theta_{S_0}} \big)^{-|\bar{\phi}_{B}|}-1
      \big] r\mathrm{d}r   \Big).
  \end{aligned}$}
 \end{equation}

\end{lemma}
\begin{proof}
  %The derivation is omitted  due to space limitation.
  Please refer to Appendix \ref{Appendix:coverage}.
\end{proof}

The   coverage probability analysis is
converted into the analysis of the higher-order derivatives  according to Lemma \ref{lemcoverage}, 
% of $Y_{I_{B0}} $ and $Y_{I_{B}} $,
and the coverage of probability is rewritten as 
    \begin{equation}\label{eq:outtolap2}
      \begin{aligned}
        p_{\mathrm{c}}(d_{00})
    %     &=\mathbb{P}[ S_{0}  > T(I_{B0}+I_{B}+\bar{I}_{A}+\sigma^2) ]  
    % \\
     &=\sum_{i=0}^{k_{S_0}-1} \frac{(-1)^i}{i!} \frac{\partial^i}{\partial^i s}
         \Big\{ L(s )  \Big\}_{s=1},
      \end{aligned}
    \end{equation}
where
\begin{equation}
  \begin{aligned}
  &L(s ) = e^{-s\frac{TI_{e}}{\theta_{S_0}}} \cdot
  \Big(
    1+s\frac{T\rho_B\beta_0 d_{00}^{-\alpha_1}}{\theta_{S_0}}\Big)^{1-|\bar{\phi}_{B}|} 
    \\ &\cdot
    \mathrm{exp}\Big( 2\pi \lambda_B \int_{d_{00}}^{\sqrt{\frac{|\mathcal{A}|}{\pi}}} \big[
      \big(1+s\frac{T \rho_B \beta_{0} r^{-\alpha_1} }{\theta_{S_0}} \big)^{-|\bar{\phi}_{B}|}-1
      \big] r\mathrm{d}r   \Big).    
    \end{aligned}
\end{equation}

% The higher-order derivatives of $L(s )$ is derived in the next part.

\subsection{Evaluation of  Higher-order Derivatives} 

The objective function $L(s )$ of  derivatives 
in   (\ref{eq:outtolap2}) is rewritten in the form of the exponential function, i.e.,
\begin{equation}
  \begin{aligned}
  L(s ) &= \mathrm{exp}\bigg\{\underbrace{-s\frac{TI_{e}}{\theta_{S_0}}}_{D_1(s)}+
  \underbrace{(1-|\bar{\phi}_{B}|) \mathrm{ln}\Big(
    1+s  T_{\theta_{S_0}} d_{00}^{-\alpha_1}  \Big)}_{D_2(s)} 
     \\ &+
     \underbrace{2\pi \lambda_B \int_{d_{00}}^{\sqrt{\frac{|\mathcal{A}|}{\pi}}} \big[
      \big(1+s T_{\theta_{S_0}}  r^{-\alpha_1}   \big)^{-|\bar{\phi}_{B}|}-1
      \big] r\mathrm{d}r }_{D_3(s)}
    \bigg\}    ,
  \end{aligned}
\end{equation}
where $T_{\theta_{S_0}}=\frac{T\rho_B\beta_0  }{\theta_{S_0}}$ is applied for convenience.

Since $ L(s )$ is a composite function of 
$g(s)=D_1(s)+D_2(s)+D_3(s)$, the special case of
Fa{\`a} di Bruno's formula with exponential functions can be applied to efficiently derive 
the $i$th order derivatives of $ L(s )$ \cite{Lyu2021IRS,johnson2002curious}, i.e.,
\begin{equation}
  % \resizebox{1\hsize}{!}
  % {$
  \begin{aligned}
    \frac{\partial^i}{\partial^i s}
         L(s ) &= \frac{\partial^i}{\partial^i s}\big\{\mathrm{exp}\big(
          g(s)\big)\big\} 
           \\ &= \mathrm{exp}\big(
            g(s)\big) B_i 
          \Big(\frac{\partial^1  g(s) }{\partial^1 s}, ...,\frac{\partial^i g(s) }{\partial^i s}\Big),
  \end{aligned}
  % $}
\end{equation}
where $B_i(x_1,...,x_i)$ denotes the $i$th complete exponential Bell polynomial,
and the coefficients can be efficiently found by 
its definition \cite{Tanbourgi2014dual,Ivanoff1958problem}.
The remaining work is to derive the higher-order derivatives of $g(s)$, which can be decomposed as
\begin{equation}
  \frac{\partial^i g(s) }{\partial^i s}=\frac{\partial^i D_1(s)}{\partial^i s}+\frac{\partial^i D_2(s)}{\partial^i s}+\frac{\partial^i D_3(s)}{\partial^i s}.
\end{equation}

For the first term $D_1(s)$, the derivatives   can be
expressed piecewise in the following form
\begin{equation}\label{eq:deltaD1}
  \frac{\partial^i D_1(s)}{\partial^i s}=\left \{ 
    \begin{aligned}
      -\frac{TI_{e}}{\theta_{S_0}}, \quad i=1\\
      0, \quad i>1.
    \end{aligned}\right 
    .
\end{equation}

Additionally, for the second term $D_2(s)$, there is 
\begin{equation}\label{eq:deltaD2}
  \frac{\partial^i D_2(s)}{\partial^i s}=(-1)^{i-1}(1-|\bar{\phi}_{B}|) (i-1)!
  \Big(\frac{T_{\theta_{S_0}} d_{00}^{-\alpha_1}}{1+T_{\theta_{S_0}} d_{00}^{-\alpha_1}s}\Big)^i.
\end{equation}

Finally, the    derivatives of the third term $D_3(s)$ is 
\begin{equation}\label{eq:deltaD3}
  \begin{aligned}
  &\frac{\partial^i D_3(s)}{\partial^i s}=\\
  &2\pi \lambda_B \int_{d_{00}}^{\sqrt{\frac{|\mathcal{A}|}{\pi}}}  
      \frac{(|\bar{\phi}_{B}|+i-1)!}{(|\bar{\phi}_{B}|-1)!} \cdot
      \frac{(-T_{\theta_{S_0}} r^{-\alpha_1})^i}{(1+T_{\theta_{S_0}} r^{-\alpha_1}s)^{|\bar{\phi}_{B}|+i}}
        r\mathrm{d}r.
    \end{aligned}
\end{equation}

Therefore, 
substituting (\ref{eq:outtolap2}) with (\ref{eq:deltaD1}), (\ref{eq:deltaD2}) and (\ref{eq:deltaD3}) 
back into (\ref{eq:cov1}) yields the network coverage probability at a distance of
$d_{00}$,  which is shown in (\ref{eq:detailed_pcd00}).

\begin{figure*}[ht] %hb代表放在文章底部，%ht为放在文章顶部
  %上面那条横线
    \begin{equation}\label{eq:detailed_pcd00}
      \begin{aligned}
        p_{\mathrm{c}}(d_{00})&=\sum_{i=0}^{k_{S_0}-1} \frac{(-1)^i}{i!} 
        \Bigg\{
          \mathrm{exp}\bigg\{-\frac{TI_{e}}{\theta_{S_0}}+
          (1-|\bar{\phi}_{B}|) \mathrm{ln}\Big(
            1+  T_{\theta_{S_0}} d_{00}^{-\alpha_1}  \Big)
            +2\pi \lambda_B \int_{d_{00}}^{\sqrt{\frac{|\mathcal{A}|}{\pi}}} \big[
              \big(1+ T_{\theta_{S_0}}  r^{-\alpha_1}   \big)^{-|\bar{\phi}_{B}|}-1
              \big] r\mathrm{d}r 
            \bigg\} \\
          &  B_i 
          \Big( -\frac{TI_{e}}{\theta_{S_0}}+ (1-|\bar{\phi}_{B}|)  
          \frac{T_{\theta_{S_0}} d_{00}^{-\alpha_1}}{1+T_{\theta_{S_0}} d_{00}^{-\alpha_1}}
          +2\pi \lambda_B \int_{d_{00}}^{\sqrt{\frac{|\mathcal{A}|}{\pi}}}  
          |\bar{\phi}_{B}|
          \frac{-T_{\theta_{S_0}} r^{-\alpha_1}}{(1+T_{\theta_{S_0}} r^{-\alpha_1})^{|\bar{\phi}_{B}|+1}}
            r\mathrm{d}r, 
          ...,\\ &
          (-1)^{i-1}(1-|\bar{\phi}_{B}|) (i-1)!
          \Big(\frac{T_{\theta_{S_0}} d_{00}^{-\alpha_1}}{1+T_{\theta_{S_0}} d_{00}^{-\alpha_1}}\Big)^i
          +2\pi \lambda_B \int_{d_{00}}^{\sqrt{\frac{|\mathcal{A}|}{\pi}}}  
          \frac{(|\bar{\phi}_{B}|+i-1)!}{(|\bar{\phi}_{B}|-1)!} \cdot
          \frac{(-T_{\theta_{S_0}} r^{-\alpha_1})^i}{(1+T_{\theta_{S_0}} r^{-\alpha_1})^{|\bar{\phi}_{B}|+i}}
            r\mathrm{d}r
            \Big)\Bigg\}.
      \end{aligned}
    \end{equation}
    \hrulefill
  \end{figure*}  

  \subsection{Improvement of the   Coverage Analysis} 

The network coverage probability has been expressed as the sum of higher-order derivatives based
on Lemma \ref{lemcoverage}. 
However, as can be seen from Appendix \ref{Appendix:coverage}, this rewriting holds 
when the shape parameter $k_{S_0}$  is a positive integer. Since $k_{S_0}$ is obtained
from (\ref{eq:apppara}) via moment matching and is not guaranteed to be a positive integer,
it can be solved by the 
linear weighting of the upper and lower integers of  $k_{S_0}$ 
with the corresponding upper and lower integers $\left\lceil k_{S_0}\right\rceil $
and $\left\lfloor k_{S_0} \right\rfloor$. Therefore, the  coverage probability  with respect to $k_{S_0}$
can be obtained from the linear weighted probability of $\left\lceil k_{S_0}\right\rceil $
and $\left\lfloor k_{S_0} \right\rfloor$, i.e.,
  \begin{equation}\label{eq:weightedpcd00}
    p_{\mathrm{c }}(d_{00}) \!=\!(\left\lceil \!k_{S_0}\!\right\rceil \!-\!
    k_{S_0})p_{\mathrm{c },\left\lfloor k_{S_0} \right\rfloor} (d_{00})
  \!+\!( 
  k_{S_0}\!-\!\left\lfloor \!k_{S_0} \!\right\rfloor )p_{\mathrm{c,\left\lceil k_{S_0}\right\rceil } 
  }(d_{00}),
  \end{equation}
  where the weights are determined by the distance of $k_{S_0}$ from the upper 
and lower integers.

Finally, note that the signal power $S_0$ follows the distribution of
the square of shifted Nakagami, which is approximated as a Gamma distribution via moment matching.
As   AP density or power  or the distance $d_{00}$ increases, 
the proportion of the desired signal $L_A$ in $S_0$ increases accordingly, 
which ultimately leads to a larger shape parameter $k_{S_0}$ in (\ref{eq:apppara}).
A larger $k_{S_0}$ means that a higher-order derivation of $L(s)$ in (\ref{eq:outtolap2})  
is required for  the analysis of coverage probability, which brings a huge 
computational burden and overhead. 
Fortunately, according to the central limit theorem, $S_0$  follows the 
normal distribution when $k_{S_0}$ is large enough.
With $ k_{S_0} =\frac{\big( \mathbb{E}[S_0] \big)^2}{ \mathrm{Var}\{S_0\} }$,
the standard deviation is much smaller than its mean. Therefore, $S_0$ can be approximated
by its mean $\bar{S}_0=k_{S_0}\theta_{S_0}$, and the coverage probability
with respect to $d_{00}$ in (\ref{eq:outtolap}) is expressed after transposition as
\begin{equation}\label{eq:outtoIBstar}
  \begin{aligned}
     p_{\mathrm{c}}(d_{00})&=\mathbb{P}[   I_{B0}+I_{B} <\frac{\bar{S}_0}{T}-\bar{I}_{A}-\sigma^2]. \\
  % & =\mathbb{P}[ I_{B}^{\star}< \frac{\bar{S}_0}{T}-\bar{I}_{A}-\sigma^2] ,
%  &= \mathbb{E}_{I_{B0},I_{B}}\Big[\sum_{i=0}^{k_{S_{\mathrm{sc}}}-1} \frac{1}{i!} 
%  \big( \frac{T(I_{B0}\!+\!I_{B}\!+\! I_{e}^{\star})}{\theta_{S_{\mathrm{sc}}}}\big)^i 
%  e^{-\frac{T(I_{B0}+I_{B}+ I_{e}^{\star} )}{\theta_{S_{\mathrm{sc}}}}}\Big] \\
%  &=\sum_{i=0}^{k_{S_{\mathrm{sc}}}-1} \frac{(-1)^i}{i!} \frac{\partial^i}{\partial^i s}
%  \Big\{ e^{-s\frac{TI_{e}^{\star}}{\theta_{S_{\mathrm{sc}}}}}
%    \mathcal{L}_{Y_{I_{B0}}}(s) 
%    \mathcal{L}_{Y_{I_{B}}}(s)   \Big\}_{s=1}
\end{aligned}
\end{equation}
% where $I_{B}^{\star}=I_{B0}+I_{B}$.

% Notice that the coverage probability in (\ref{eq:outtoIBstar}) has the same 
% structure as that in (\ref{eq:outtolap}), and the analysis and derivation for (\ref{eq:outtolap})
% still apply to the shifted signal power $S_{\mathrm{sc}}$.
% Thus, after the signal power shift,  the complete expression for the network coverage 
% probability with respect to $d_{00}$  can be obtained by bringing the 
% shifted signal parameters $k_{S_{\mathrm{sc}}}$, $\theta_{S_{\mathrm{sc}}}$ 
% and $I_{e}^{\star}$ into (\ref{eq:detailed_pcd00}). 
% Finally, the average  coverage probability in  the network is obtained 
% by substituting the linear weighted result of (\ref{eq:detailed_pcd00}) back into (\ref{eq:cov1}).

The coverage probability in (\ref{eq:outtoIBstar})  is transformed based on the 
CDF of 
$I_{B0} + I_{B} $, whose joint probability distribution  is difficult to derive.
Fortunately, this  issue can be addressed by considering the mutual 
independence of $I_B$ and 
$I_{B0}$ for a given $d_{00}$ with Lemma \ref{lemconditionalp}.

\begin{lemma}
  \label{lemconditionalp}

For a large
value of $k_{S_0}$,
the network coverage probability with $d_{00}$ can be expressed according to   $I_{e}^{\star}$ as
\begin{equation}
  \resizebox{1\hsize}{!}{$
  p_{\mathrm{c}}(d_{00})\!=\!\left \{ 
    \begin{aligned}
      \int_{0}^{I_{e}^{\star}} \!
      f_{I_{B0}}(y)
      \mathcal{L}^{-1}\Big\{ \frac{1}{s} 
      \mathcal{L}_{I_B|d_{00}}(s)
       \Big\}(I_{e}^{\star}\!-\!y) \mathrm{d} y
      ,&  
    I_{e}^{\star}\!>\! 0, 
    \\ 
    0, \quad\quad\quad\quad\quad\quad\quad\quad\quad\quad\quad\quad\quad\quad\quad&
    I_{e}^{\star}\!\leq\! 0    , 
    \end{aligned}\right . $}
\end{equation}
where $I_{e}^{\star}=\frac{\bar{S}_0}{T}-\bar{I}_{A}-\sigma^2$, 
the PDF of $I_{B0}$ is 
\begin{equation}
  f_{I_{B0}}(y)=\frac{1}{\Gamma(|\bar{\phi}_{B}|-1) 
  (\rho_B \beta_{00})^{|\bar{\phi}_{B}|-1}  }y^{|\bar{\phi}_{B}|-2}
  e^{\frac{-y}{\rho_B \beta_{00}}},
\end{equation}
and the Laplace transform of $I_B$ is 
\begin{equation}\label{eq:lapIB}
  \resizebox{1\hsize}{!}{$
  \begin{aligned}
   & \mathcal{L}_{I_B|d_{00}}(s)\!=\!
  \mathrm{exp}\big\{
      2\pi \lambda_B \int_{d_{00}}^{\sqrt{\frac{|\mathcal{A}|}{\pi}}}
      [(1\!+\!s\rho_B \beta_0 r^{-\alpha_1})^{-|\bar{\phi}_{B}|}\!-\!1]r \mathrm{d} r
      \big\}  ,
  \end{aligned}$}
\end{equation}
whose inverse Laplace transform can be calculated directly with Matlab.
% $k_{I_{B }^{\star}}$ and $\theta_{I_{B }^{\star}}$ are the shape and scale 
% parameters of the Gamma approximation of $I_{B }^{\star}=I_{B0}+I_{B}$ and 
% are derived in Appendix \ref{Appendix:conditionalp}.

\end{lemma}
\begin{proof}
  Please refer to Appendix \ref{Appendix:conditionalp} for detailed derivation.
\end{proof}

% \subsection{Special Case with $|\bar{\phi}_B|=1$}  
% Consider a small cell scenario where the number of UE is comparable to the 
% number of BSs. Suppose that in this special case there is $|\bar{\phi}_B|=1$.
% From (\ref{eq:detailed_pcd00}), the conditional coverage probability 
% $ P_{\mathrm{c }} $ is further simplified  in (\ref{eq:detailed_pcd00_case})
% at the top of this page.

% \begin{figure*}[ht] %hb代表放在文章底部，%ht为放在文章顶部
%   %上面那条横线
%     \begin{equation}\label{eq:detailed_pcd00_case}
%       \begin{aligned}
%         P_{\mathrm{c }} &=\sum_{i=0}^{k_{S_0}-1} \frac{(-1)^i}{i!} 
%         \Bigg\{
%           \mathrm{exp}\bigg\{-\frac{TI_{e}}{\theta_{S_0}}+
%             +2\pi \lambda_B \int_{d_{00}}^{\sqrt{\frac{|\mathcal{A}|}{\pi}}} \big[
%               \big(1+ T_{\theta_{S_0}}  r^{-\alpha_1}   \big)^{-1}-1
%               \big] r\mathrm{d}r 
%             \bigg\}  \times\\
%           &  B_i 
%           \Big( -\frac{TI_{e}}{\theta_{S_0}}
%           +2\pi \lambda_B \int_{d_{00}}^{\sqrt{\frac{|\mathcal{A}|}{\pi}}}  
%           \frac{-T_{\theta_{S_0}} r^{-\alpha_1}}{(1+T_{\theta_{S_0}} r^{-\alpha_1})^{2}}
%             r\mathrm{d}r, 
%           ..., 
%           2\pi \lambda_B \int_{d_{00}}^{\sqrt{\frac{|\mathcal{A}|}{\pi}}}  
%           i!
%           \frac{(-T_{\theta_{S_0}} r^{-\alpha_1})^i}{(1+T_{\theta_{S_0}} r^{-\alpha_1})^{1+i}}
%             r\mathrm{d}r
%             \Big)\Bigg\}.
%       \end{aligned}
%     \end{equation}
%     \hrulefill
%   \end{figure*}  
\section{Network Average  Downlink Rate}
This section focuses on the achievable rate of 
HCCNs in downlink transmission.
Note that the unit of the  downlink rate   
is nat/sec/Hz.
The rate analysis   is  based on   the 
coupling relationship between the aggregated   signal and  interference.

With the SINR expression  in (\ref{eq:SINR}), the basic form of the average achievable rate 
$R$  is first  written as 
\begin{equation}\label{eq:Rinitial}
  \begin{aligned}
  R&=\mathbb{E}\big[\mathrm{ln}(1+\Omega) \big] 
  =  \int_{0}^{\sqrt{\frac{|\mathcal{A}|}{\pi}}} R(r) f_{d_{00}}(r) 
  \mathrm{d}r,
\end{aligned}
\end{equation}
where the distance $d_{00}$ between  UE 0 and the nearest 
BS 0 is treated as a random variable, and $f_{d_{00}}(r) $ is obtained from (\ref{eq:fd00}).
The average achievable rate $R(d_{00})$  is
\begin{equation}\label{eq:Rd00initial}
  R(d_{00})= \Big\{
   \mathbb{E}\big[\mathrm{ln}(1+\Omega) \big]\Big\}_{d_{00}}.
\end{equation}

%\subsection{Coupling Analysis Based Approach}
% To reduce the computational complexity, the coupling between the  signal 
% and the interference can be considered for the derivation of the network average rate.
Since the average achievable   rate $R(d_{00}) $ 
 in (\ref{eq:Rd00initial})  is in the form of 
ln function, 
the following transform is first introduced
\begin{equation}
  \mathrm{ln}(1+x)=\int_{0}^{\infty} \frac{e^{-y}}{y}(1-e^{-xy})\mathrm{d}y.
\end{equation} 

Therefore, the average achievable rate with $d_{00}$ is accordingly rewritten as 
\begin{equation}\label{eq:Eln1omega}
  \resizebox{1\hsize}{!}{$
  \begin{aligned}
    &\Big\{
      \mathbb{E}\big[\mathrm{ln}(1+\Omega) \big]\Big\}_{d_{00}}  =
    \mathbb{E}\big[\int_{0}^{\infty}  \frac{e^{-y}}{y} (1-e^{-\Omega y})\mathrm{d}y \big]\\
    &\overset{(a)}{=}\mathbb{E}\big[ \int_{0}^{\infty}  
    \frac{e^{-s(I_B+I_{B0}+\bar{I}_{A}+\sigma^2)}}{s} (1-e^{-S_0 s}) \mathrm{d}s \big] \\
    &\overset{(b)}{=}\int_{0}^{\infty}  \mathbb{E}\big[
    \frac{e^{-s(I_B+I_{B0}+\bar{I}_{A}+\sigma^2)}}{s} (1-e^{-S_0 s}) \big]\mathrm{d}s \\
    &\overset{(c)}{=}\int_{0}^{\infty}    \frac{e^{-s(  \bar{I}_{A}+\sigma^2)}}{s}
    \mathbb{E}\big[e^{-s I_B } \big]
    \mathbb{E}\big[e^{-s I_{B0} }-e^{-s(I_{B0}+S_0)} \big] \mathrm{d}s \\
    & =\int_{0}^{\infty}    \frac{e^{-s(  \bar{I}_{A}+\sigma^2)}}{s}
    \mathcal{L}_{I_{B}|d_{00}}(s)  
   \big[\mathcal{L}_{I_{B0}|d_{00}}(s) \!-\!\mathcal{L}_{I_{B0}+S_0|d_{00}}(s)  \big]\mathrm{d}s. 
  \end{aligned}$}
\end{equation}
 
In (\ref{eq:Eln1omega}), relation (a) is obtained from the variable substitution
$y=s(I_B+I_{B0}+\bar{I}_{A}+\sigma^2)$. Relation (b) changes the order 
of integral and expectation with Fubini theorem.
Relation (c) comes from the analysis above that 
$I_{B0}$ and $I_B$ are independent of 
each other with a given $d_{00}$.

First,   the Laplace transform $ \mathcal{L}_{I_{B}|d_{00}}(s)$ in (\ref{eq:Eln1omega})
is shown in (\ref{eq:lapIB}) with the detailed derivation shown in Appendix \ref{Appendix:conditionalp}.
 
Second, based on the approximation in (\ref{eq:IB0-1}) and the 
moment generating function   for  Gamma variable in (\ref{eq:mgf}), 
the Laplace transform $\mathcal{L}_{I_{B0}|d_{00}}(s)$ is  
\begin{equation}\label{eq:esIB0}
  \begin{aligned}
    \mathcal{L}_{I_{B0}|d_{00}}(s)&=\mathbb{E}\big[e^{-s \rho_B   \beta_{00}\kappa_{B,0} }\big]\\
    &=(1+s \rho_B   \beta_{0} d_{00}^{-\alpha_1} )^{1-|\bar{\phi_B}|}.
  \end{aligned}
\end{equation}

Finally,  since both the signal $S_0$  and the interference $I_{B0}$ rely on 
the channel $\mathbf{h}_{00}$ between UE 0 and BS 0,  they are   not independent of 
each other. 
Fortunately, by denoting the sum of signals and  interference as $S_{I}=S_0+I_{B0} $
and matching the moments of $S_{I}$, the following Lemma is then  introduced.
\begin{lemma}
  \label{lemSI}
  According to the definition of the Gamma random variable \cite{pishro2014introduction}, the 
  sum of   signals and  interference   $S_{I}=S_0+I_{B0} $ 
  can be approximated as the following Gamma distribution 
  $\Gamma(k_{S_I},\theta_{S_I})$ with
  \begin{equation}\label{eq:appparaksi}
    \begin{aligned}
      k_{S_I}&=\frac{\big( \mathbb{E}[S_I] \big)^2}{ \mathrm{Var}\{S_I\} }
      =\frac{\big( \mathbb{E}[S_I] \big)^2}{ \mathbb{E}[S_I^2] -\big( \mathbb{E}[S_I] \big)^2},\\
      \theta_{S_I}&=\frac{ \mathrm{Var}\{S_I\} }{\mathbb{E}[S_I] }
      =\frac{ \mathbb{E}[S_I^2] -\big( \mathbb{E}[S_I] \big)^2}{\mathbb{E}[S_I] },
    \end{aligned}
  \end{equation}
where  the first- and second-order moments of $S_I$ are 
\begin{equation}\label{eq:ES_I}
  \mathbb{E}[S_I]=\rho_B  \beta_{00} ( N_B +|\bar{\phi}_{B}|-1)
  +2\rho_B^{\frac{1}{2}}\beta_{00}^{\frac{1}{2}}L_A\frac{\Gamma(N_B +\frac{1}{2})}{\Gamma(N_B)}  
  +L_{A}^{2}  ,
\end{equation}
\begin{equation}\label{eq:ES_I2}
  \begin{aligned}
  \mathbb{E}[S_I^2] & =\rho_B^2  \beta_{00}^2((N_B+|\bar{\phi}_B|)^2+|\bar{\phi}_B|-N_B-2)\\ 
  &+ 4\rho_B^\frac{3}{2}\beta_{00}^{\frac{3}{2}} L_A
  \frac{\Gamma(N_B +\frac{3}{2})}{\Gamma(N_B)}(1+\frac{(|\bar{\phi}_B|-1)}{N_B}) \\
  & +\rho_B\beta_{00}L_A^2 (6N_B+2|\bar{\phi}_B|-2)\\
  &+4\rho_B^{\frac{1}{2}}\beta_{00}^{\frac{1}{2}}L_A^3\frac{\Gamma(N_B +\frac{1}{2})}{\Gamma(N_B)}+L_{A}^{4}.
\end{aligned}\end{equation}
\end{lemma}

\begin{proof}
  %The derivation is omitted  due to space limitation.
  Please refer to Appendix \ref{Appendix:SI}.
\end{proof}

The Laplace transform $\mathcal{L}_{ S_I|d_{00}}(s) $
with given $d_{00}$
is approximated based on  Lemma \ref{lemSI} as 
\begin{equation}\label{eq:esIB0S0}
  \begin{aligned}
    \mathcal{L}_{S_I|d_{00}}(s) &=\mathbb{E}\big[e^{-s S_I }\big] 
     =(1+s\theta_{S_I})^{-k_{S_I}}.
  \end{aligned}
\end{equation}

The Laplace transforms of the components in (\ref{eq:Eln1omega}) are calculated in 
(\ref{eq:lapIB}), (\ref{eq:esIB0}) and (\ref{eq:esIB0S0}), respectively. 
After that, the complete expression for the achievable rate $R$ is
obtained by substituting (\ref{eq:Eln1omega}) into (\ref{eq:Rinitial}), 
displayed in (\ref{eq:detailed_R2}).
\begin{figure*}[ht] %hb代表放在文章底部，%ht为放在文章顶部
  %上面那条横线
    \begin{equation}\label{eq:detailed_R2}
      \begin{aligned}
        &R =
        \int_{0}^{\sqrt{\frac{|\mathcal{A}|}{\pi}}}
        R(d_{00})
        \frac{2\pi\lambda_B  d_{00 } e^{-\lambda_B \pi d_{00 }^2}}{p_{\mathrm{area}}}
  \mathrm{d}d_{00}\\
        &=
        \int_{0}^{\sqrt{\frac{|\mathcal{A}|}{\pi}}}
        \int_{0}^{\infty}    \frac{e^{-s(  \bar{I}_{A}+\sigma^2)}}{s}
        \mathbb{E}\big[e^{-s I_B } \big]
        \mathbb{E}\big[e^{-s I_{B0} }-e^{-s(I_{B0}+S_0)} \big]
            \frac{2\pi\lambda_B  d_{00 } e^{-\lambda_B \pi d_{00 }^2}}{p_{\mathrm{area}}}
             \mathrm{d}d_{00}\\
             &= \int_{0}^{\sqrt{\frac{|\mathcal{A}|}{\pi}}}
             \int_{0}^{\infty}    \frac{e^{-s(  \bar{I}_{A}+\sigma^2)}}{s}
             \mathrm{exp}\Big\{2\pi\lambda_B  
      \int_{d_{00}}^{\sqrt{\frac{|\mathcal{A}|}{\pi}}} 
      \big[ (1+s\rho_B\beta_{0 }r^{-\alpha_1})^{-|\bar{\phi_B}|}-1  \big] r \mathrm{d}r\Big\}
      \\ & \times
             \bigg((1+s \rho_B   \beta_{0} d_{00}^{-\alpha_1} )^{1-|\bar{\phi_B}|}
             -(1+s\theta_{S_I})^{-k_{S_I}}\bigg)\mathrm{d}s
                 \frac{2\pi\lambda_B  d_{00 } e^{-\lambda_B \pi d_{00 }^2}}{p_{\mathrm{area}}}
                  \mathrm{d}d_{00}
      \end{aligned}
    \end{equation}
    \hrulefill
  \end{figure*}

\section{Numerical Results}
In this section, the analytical results of the coverage probability of 
the HCCN are verified 
by the MC simulation results.
Each result of the MS simulation is averaged from 2500
randomly generated wireless node distributions.
All nodes are randomly distributed in a circular area of radius 500m.
UE 0 
is located at the center of the circle. The general simulation parameters
are shown in Table \ref{table1p}, if not stated otherwise.

%经典三线表
\begin{table}[H]
  \caption{Simulation parameters}%标题
  \label{table1p}
  \centering%把表居中
  \begin{tabular}{cccc}%四个c代表该表一共四列，内容全部居中
  \toprule%第一道横线
  Parameter&Value&Parameter&Value \\
  \midrule%第二道横线 
  $\lambda_B$&40$/\mathrm{km}^2$&$\lambda_A$&$200/\mathrm{km}^2$\\
  $\lambda_U$&$120/\mathrm{km}^2$&$\alpha_1$&2.8 \\
  $\alpha_2$&1.5&$\frac{P_B}{\sigma^2}$&130dB \\
  $P_B$&50dBm&$P_A$&10dBm \\
  $N_B$&8&$N_A$&2 \\
  $f$&3.5GHz&$C$&$3\times10^8$m/s \\
  $\beta_0$&$(\frac{C}{4\pi f})^2$&$\delta_0$&$(\frac{C}{4\pi f})^2$\\
  \bottomrule%第三道横线
  \end{tabular}
  \end{table}

% \begin{figure}[!t]
%   \centering
%     {\includegraphics[width=0.72\columnwidth]{./figure/LA-changeU.eps}}%
%   \caption{ Value of $S_{0A}$ and its approximation $L_A$ under different UE density $\lambda_U$.}
%   \label{fig:LA-U}
% \end{figure}

% Fig. \ref{fig:LA-U} shows the value of 
% $\mathbb{E}[S_{0A}]$ and $L_A$ with different UE density $\lambda_U$, 
% where $\alpha_2=1.2$, $N_A= 4$.

% \begin{figure}[!t]
%   \centering
%     {\includegraphics[width=0.72\columnwidth]{./figure/IA-changeU.eps}}%
%   \caption{ Value of $\bar{I}_{A}$  under different UE density $\lambda_U$.}
%   \label{fig:IA-U}
% \end{figure}
% Fig. \ref{fig:LA-U} shows the value of 
% $\bar{I}_{A}$ with different UE density $\lambda_U$, 
% where $\alpha_2=1.2$, $N_A= 4$. Since there is 
% \begin{equation}
%   \begin{aligned}
%     \bar{I}_{A}&=\frac{2\pi \rho_A  \lambda_A \delta_{0}(\lambda_U|\mathcal{A}|-1)}
%     {2-\alpha_2}
%     \Big(\frac{|\mathcal{A}|}{\pi}\Big)^{1-\frac{\alpha_2}{2}} 
%     \\ & =\frac{2\pi P_A  \lambda_A \delta_{0}(1-\frac{1}{\lambda_{U}|\mathcal{A}|})}
%     {2-\alpha_2}
%     \Big(\frac{|\mathcal{A}|}{\pi}\Big)^{1-\frac{\alpha_2}{2}} ,
%   \end{aligned}
% \end{equation}
% $\bar{I}_{A}$ will increase slightly with UE density $\lambda_U$.

% \begin{figure}[!t]
%   \centering
%     {\includegraphics[width=0.72\columnwidth]{./figure/map-111213-Tpcov-2.eps}}%
%   \caption{ Coverage probability of different architectures under different $T$.}
%   \label{fig:Tpcov2}
% \end{figure}
 
\begin{figure}[!t]
  \centering
    {\includegraphics[width=0.72\columnwidth]{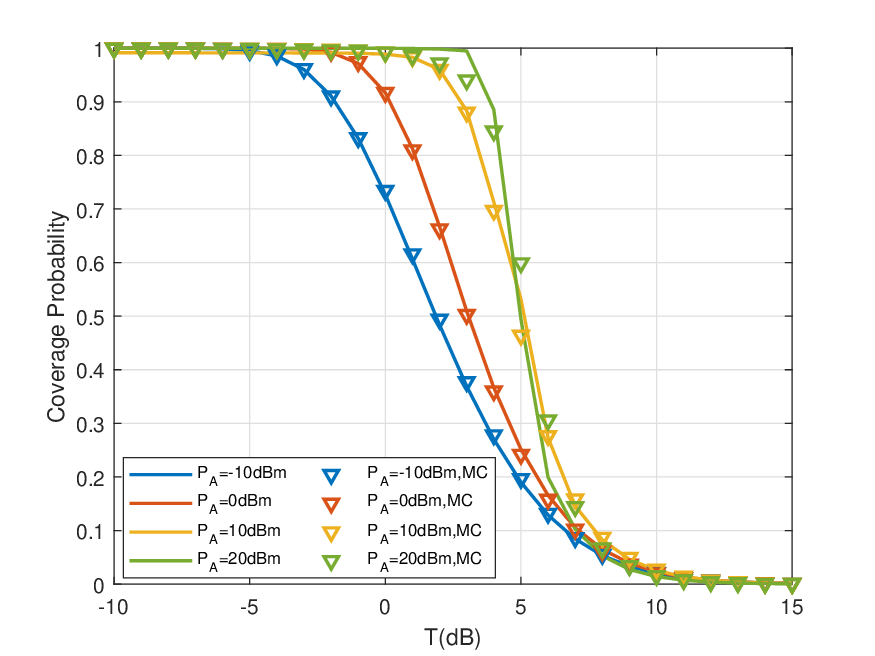}}%
  \caption{ Coverage probability under different $P_A$ versus SINR threshold $T$.}
  \label{fig:cov_PA_Tdb_MC}
\end{figure}
 
The network coverage probability under different AP power $P_A$ versus SINR threshold $T$
is shown in Fig. \ref{fig:cov_PA_Tdb_MC}. Overall, the  coverage probability $p_c$ 
decreases as the threshold $T$ rises. 
As AP power $P_A$  increases, the  coverage performance at 
different SINRs is significantly improved, especially at low SINRs.
The slope of the coverage probability also becomes larger as the power PA increases.
This is consistent with the expectation, 
as the cell-free APs effectively serve UEs at the cellular edge and provide UEs with 
uniform good service.
Meanwhile,  the coverage probability  of high SINRs first increases and 
then decreases with the increase of $P_A$. This is  because 
the contribution of the desired signal dominates when $P_A$ is low, 
and when $P_A$ reaches a certain level (e.g., 10dBm), the interference dominates,
and further increasing $P_A$ will lead to performance degradation due to the strong 
interference.

\begin{figure}[!t]
  \centering
    {\includegraphics[width=0.72\columnwidth]{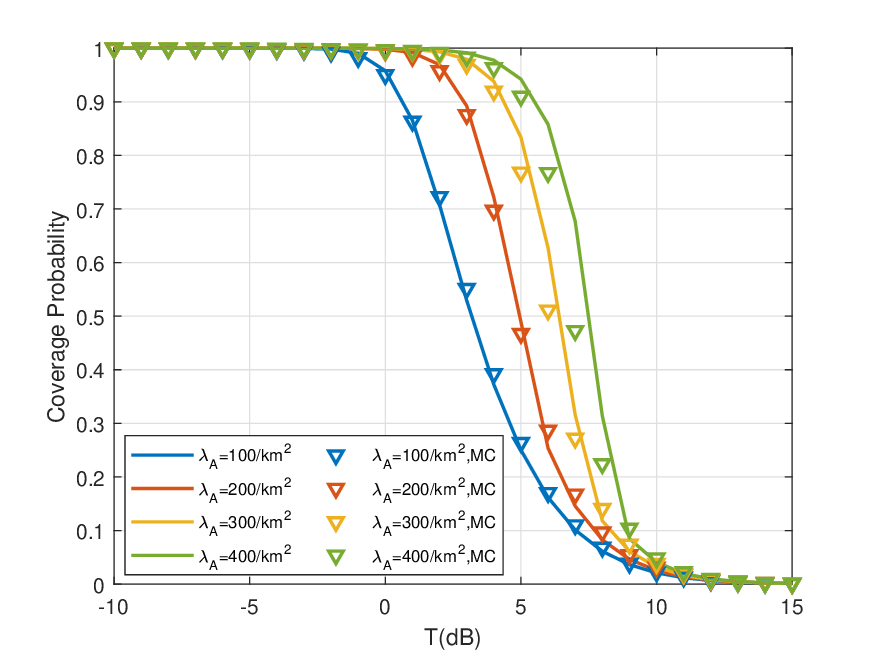}}%
  \caption{ Coverage probability under different $\lambda_A$ versus SINR threshold $T$.}
  \label{fig:cov_lambdaA_Tdb_MC}
\end{figure}

Fig. \ref{fig:cov_lambdaA_Tdb_MC} shows the network coverage probability 
under different AP density $\lambda_A$ with   $P_A=10$dBm.
It is observed that the coverage probability performance  decreases with the increase of $T$. 
On the other hand, increasing the density $\lambda_A$ of distributed APs can 
effectively improve the coverage probability under any $T$. 
This trend is expected, as higher density $\lambda_A$ brings closer 
UE distance and richer macro diversity, thus improving  
the  coverage probability $p_c$.

\begin{figure}[!t]
  \centering
    {\includegraphics[width=1\columnwidth]{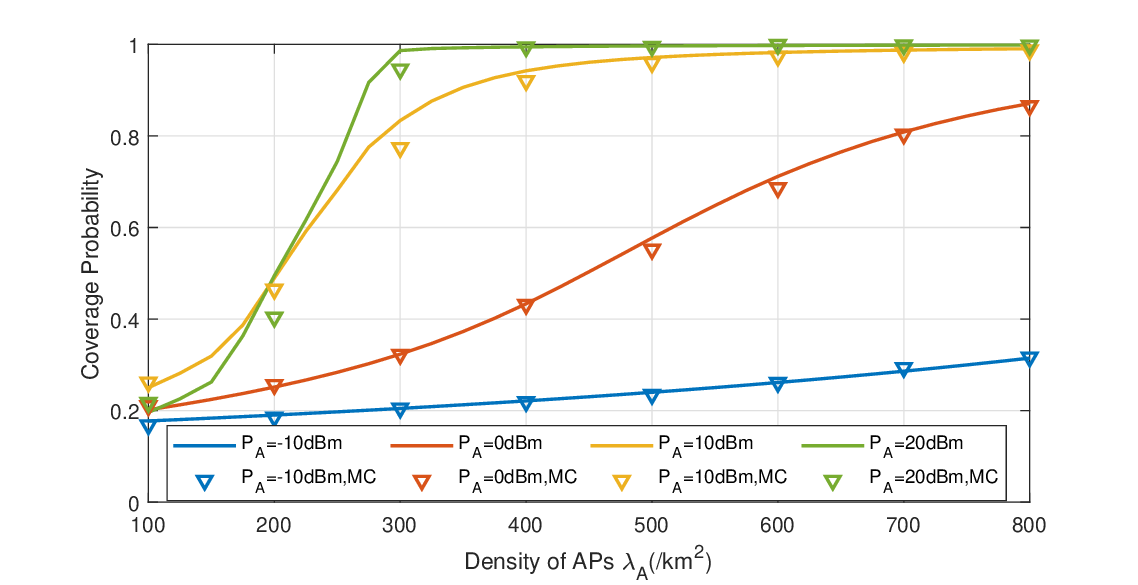}}%
  \caption{ Coverage probability under different $P_A$ versus AP density $\lambda_A$.}
  \label{fig:cov_PA_lambdaA_Pcov_Tdb5}
\end{figure}

The trend of the coverage probability with AP density  $\lambda_A$ at different 
AP powers   is shown in Fig. \ref{fig:cov_PA_lambdaA_Pcov_Tdb5}. The SINR threshold is 
$5\rm{dB}$.  On the whole, for different $P_A$, its coverage probability increases with 
the increase of $\lambda_A$, which is consistent with the analysis of 
Fig. \ref{fig:cov_lambdaA_Tdb_MC}. 
When   $P_A=-10$dBm, the increase of $\lambda_A$ has little increase in the coverage probability. 
This can be explained since when the AP power is small enough or even close to 0, 
the desired signal and interference come mainly from the cellular BSs, 
and therefore increasing $\lambda_A$ has a very limited effect.
Besides, for different densities $\lambda_A$, the coverage probability does not 
monotonically increase with the power $P_A$.
This is consistent with the analysis of Fig. \ref{fig:cov_PA_Tdb_MC} and 
provides good guidance for HCCN deployment,   i.e., the  power $P_A$ that 
maximizes the coverage probability should be efficiently selected based on the $\lambda_A$.

%\begin{figure}[!t]
%  \centering
%    {\includegraphics[width=0.72\columnwidth]{./figure/rate_PA_lambdaA.eps}}%
%  \caption{  Average achievable rate under different $P_A$ versus AP density $\lambda_A$.}
%  \label{fig:rate_PA_lambdaA_120}
%\end{figure}
 
Fig. \ref{fig:rate_PA_lambdaA_120} shows the   achievable rate  
under different   power $P_A$.
It is evident   that for different $P_A$, the increase in   $\lambda_A$ 
can help improve the   achievable rate.
The slope of the rate at $P_A=30\rm{dBm}$ is significantly larger than at $P_A=-10\rm{dBm}$, 
indicating that increasing the $\lambda_A$ tends to work better for 
situations with higher $P_A$. 
In addition, for $\lambda_A$ from 100 to 600$/\rm{km}^2$, 
$P_A=10\rm{dBm}$ and $P_A=20\rm{dBm}$ obtained similar average network rates and 
both outperformed $P_A=30\rm{dBm}$. This is expected and consistent with 
the analysis of Fig. \ref{fig:cov_PA_Tdb_MC} and Fig. \ref{fig:cov_PA_lambdaA_Pcov_Tdb5}. 
Although  HCCN with higher $P_A$ is able to serve more UEs with low SINR, 
it is not necessarily better in other performance metrics, such as network achievable 
rate.

\begin{figure}[!t]
  \centering
    {\includegraphics[width=0.72\columnwidth]{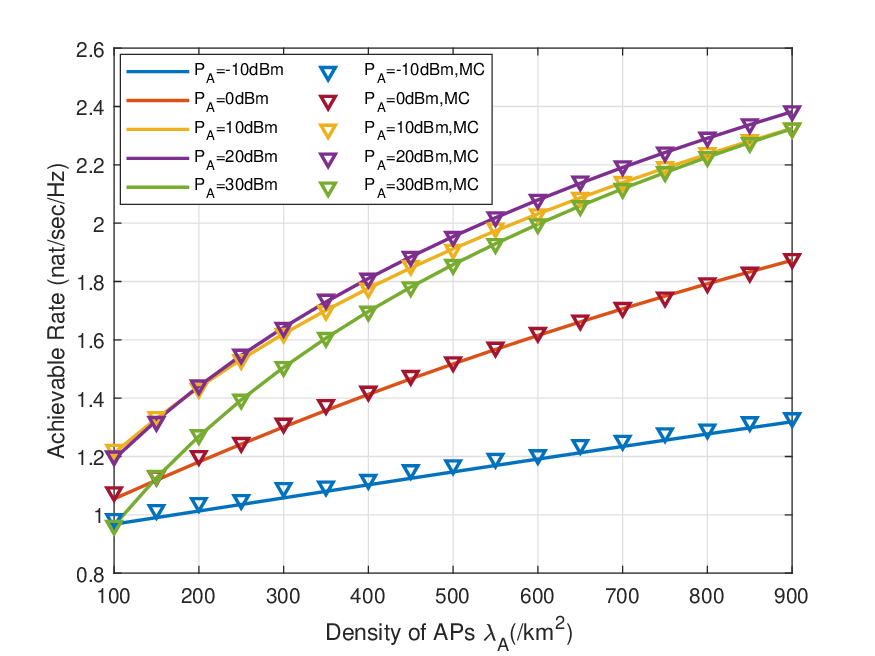}}%
  \caption{  Average achievable rate under different $P_A$ versus AP density $\lambda_A$.}
  \label{fig:rate_PA_lambdaA_120}
\end{figure}

\begin{figure}[!t]
  \centering
    {\includegraphics[width=0.72\columnwidth]{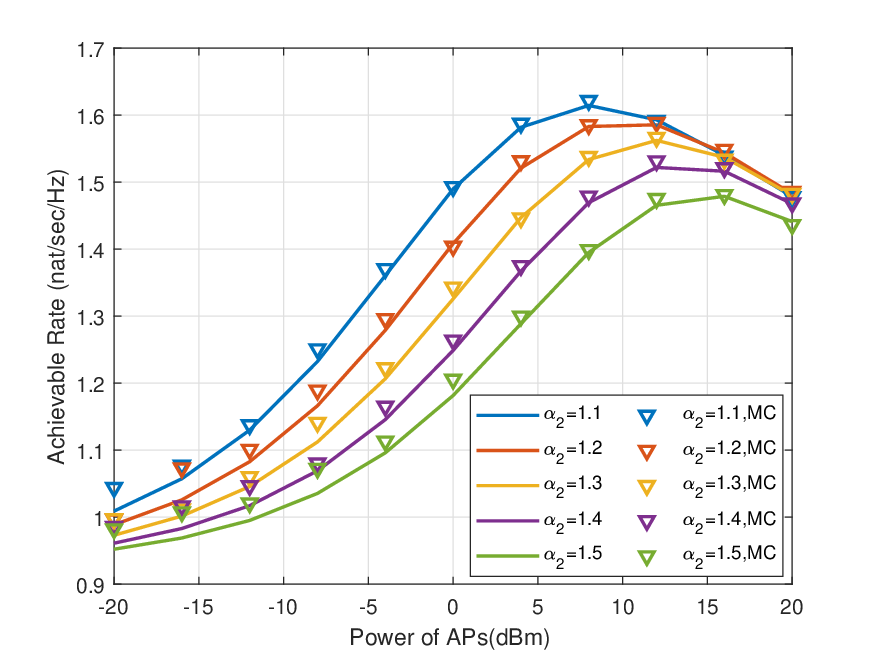}}%
  \caption{ Average achievable rate under different $\alpha_2$ versus AP power $P_A$.}
  \label{fig:rate_alpha2_PA_withMC}
\end{figure}

The  achievable rate under different path-loss
exponents $\alpha_2$ is shown in Fig. \ref{fig:rate_alpha2_PA_withMC}.
The power $P_A $ ranges from $-20\rm{dBm}$ to $20\rm{dBm}$.
A smaller $\alpha_2$ leads to higher achievable rates, 
regardless of the power $P_A$. This is expected since  
a smaller $\alpha_2$ results in  quality wireless channels with less attenuation.
Moreover, consistent with what is observed from Fig. \ref{fig:rate_PA_lambdaA_120}, 
for any $\alpha_2$, the achievable rate first rises and then falls with the increase in $P_A$.
The power $P_A$ required to achieve the maximum  rate decreases as $\alpha_2$ decreases.
This can be exploited to realize efficient HCCN deployments, 
where the network throughput can be maximized by choosing the proper $P_A$ 
to balance the  signal and interference from cell-free APs.

\begin{figure}[!t]
  \centering
    {\includegraphics[width=0.72\columnwidth]{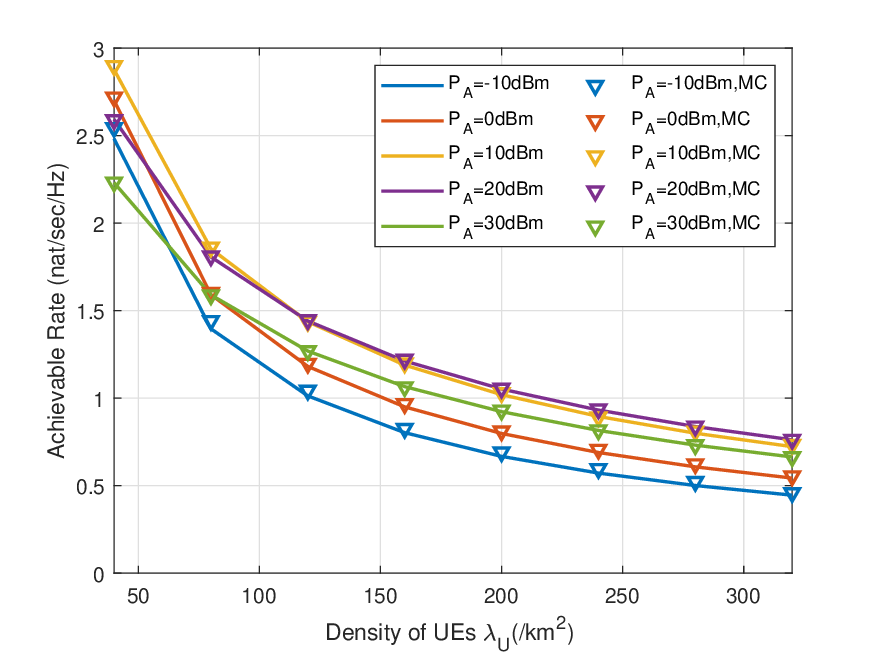}}%
  \caption{  Average achievable rate under different $P_A$ versus UE density $\lambda_U$.}
  \label{fig:rate_PA_lambdaU_120}
\end{figure}

 The network achievable rate
versus UE density $\lambda_U$ is shown in Fig. \ref{fig:rate_PA_lambdaU_120} and 
Fig. \ref{fig:rate_lambdaA_lambdaU_120}.  The AP power $P_A$ in Fig. \ref{fig:rate_PA_lambdaU_120}
ranges from $-10$dBm to $30$dBm, while the AP density $\lambda_A$ in Fig. \ref{fig:rate_lambdaA_lambdaU_120}
ranges from $100/\mathrm{km}^2$ to $500/\mathrm{km}^2$.
The average achievable rate  decreases monotonically as $\lambda_U$ increases. 
This is intuitive since increasing the  UE density inevitably reduces the power from BS/AP per UE.
From Fig. \ref{fig:rate_PA_lambdaU_120}, the achievable rate first rises and then falls 
versus $P_A$. In addition, the power to achieve the maximum rate is also affected by $\lambda_U$.
From Fig. \ref{fig:rate_lambdaA_lambdaU_120},
 a similar conclusion can be drawn as in 
Fig. \ref{fig:cov_PA_lambdaA_Pcov_Tdb5} and Fig. \ref{fig:rate_PA_lambdaA_120}, 
i.e., deploying more APs to increase $\lambda_A$ brings higher macro diversity 
and  always helps to improve the network 
performance, both in terms of coverage probability and achievable rate.

\begin{figure}[!t]
  \centering
    {\includegraphics[width=0.72\columnwidth]{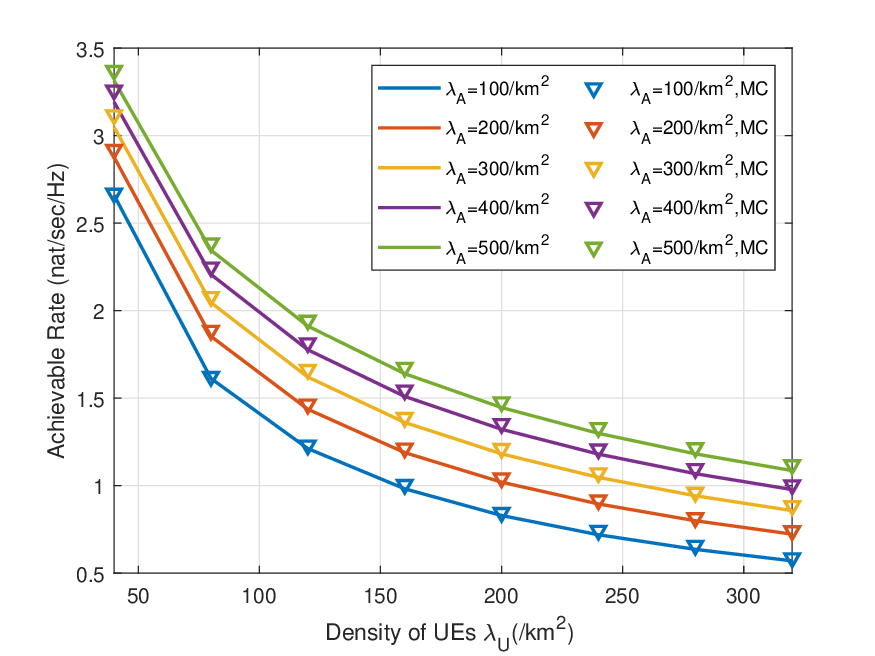}}%
  \caption{  Average achievable rate under different $\lambda_A$ versus UE density $\lambda_U$.}
  \label{fig:rate_lambdaA_lambdaU_120}
\end{figure}

% Fig. \ref{fig:Tpcov2} shows the coverage probability in   traditional cellular networks ($P_A=0$), 
% cell-free networks ($P_B=0$), and HCCNs for different SINR thresholds $T$.
% From Fig. \ref{fig:Tpcov2},   the coverage probability analysis of the hybrid network 
% is generally consistent with the results of MC simulation under different SINR threshold, 
% and can also be applied to the special case where $P_A=0$.
% First, compared with cellular networks, hybrid networks effectively improve 
% the communication performance of edge UE and reduce performance differences 
% between UE. Second, compared with cell-free networks, hybrid networks can 
% achieve better peak performance.
% Therefore, by deploying low-power APs, HCCNs can provide 
% UE with uniformly good communication services while obtaining good peak SINR performance.

\section{Conclusion}
In this paper,  the network coverage probability and average achievable rate of HCCNs are 
characterized based on the  stochastic geometry. 
By analyzing the statistical characteristics of BS and AP channels, 
the aggregated desired signal received from the HCCN is approximated by moment matching.
Further, the coverage probability of the HCCN is obtained by deriving the higher-order 
derivatives of the Laplace transforms for different interference components. 
In addition, by analyzing the coupling relationship between the desired signal 
and interference, the Laplace transforms of the interference components and the sum of intra-cell interference and signal are 
derived, and the achievable rate of HCCN is obtained.
MC simulation validates the analysis of coverage probability and achievable rate
and demonstrates that by deploying cell-free APs, 
HCCNs can narrow the service performance gap between UEs and  
effectively improve system performance.

% Moment matching is used to approximate the aggregate signal received from 
% the hybrid network to address the difficulty of distribution analysis 
% due to conjugate beamforming.
% The coverage probability  is then obtained 
% by the Laplace transform for interference.
% The analysis of coverage probability
% is validated by MC simulation, demonstrating that hybrid networks 
% can reduce the  performance gap while improving peak performance.

\section{Acknowledgment}

 \begin{appendices} 

\section{ Derivation of $ \mathbb{E}[S_0]$ and $ \mathbb{E}[S_0^2]$}\label{Appendix:moment}
 
For the desired signal power $S_0$, with given $d_{00}$ and Lemma \ref{lemassLA}, the first moment is derived as 
\begin{equation} \label{eq:es0der}
  \begin{aligned}
  \mathbb{E}[S_0] 
  % &=\mathbb{E}\big[(\sqrt{P_B \eta_B}\|\mathbf{h}_{00}\| +L_A )^2\big]_{d_{00}}\\
  % &=\mathbb{E}\big[ \rho_B|\mathbf{h}_{00}|^2 +
  % 2\rho_B^{\frac{1}{2}}L_A\|\mathbf{h}_{00}\|+L_A^2   \big]_{d_{00}}\\
  &=\rho_B\mathbb{E}\big[ |\mathbf{h}_{00}|^2\big]_{d_{00}}
  + 2\rho_B^{\frac{1}{2}}L_A\mathbb{E}\big[\|\mathbf{h}_{00}\|\big]_{d_{00}}+L_A^2\\
  &\overset{(a)}{=}\rho_B N_B \beta_{00} 
  +2\rho_B^{\frac{1}{2}}L_A\frac{\Gamma(N_B +\frac{1}{2})}{\Gamma(N_B)} \beta_{00}^{\frac{1}{2}} 
  +L_{A}^{2}.
  \end{aligned}
\end{equation}
In (\ref{eq:es0der}), there is $\rho_B=P_B \eta_B$  for simplicity.
Note that the related distributions of $\mathbf{h}_{00}$ are $|\mathbf{h}_{00}|^2\sim 
\Gamma(N_B,\beta_{00})$ and $\|\mathbf{h}_{00} \| \sim \mathrm{Nakagami}(N_B, N_B \beta_{00}).$
Since  the mean of  $X$ with $X\sim \Gamma(k,\theta)$
is $\mathbb{E}[X]=k\theta$, and the mean of $Y$ with
$Y\sim \mathrm{Nakagami}(m,\omega)$ is 
$\mathbb{E}[Y]=\frac{\Gamma(m+\frac{1}{2})}{\Gamma(m)} 
\big(\frac{\omega}{m} \big)^{\frac{1}{2}}$ \cite{pishro2014introduction},
the relation (a) is obtained accordingly.

For the second-order moment of $S_0$, there is 
\begin{equation}
  \begin{aligned}
  \mathbb{E}[S_0^2]  
  % &=\mathbb{E}\big[( \rho_B|\mathbf{h}_{00}|^2 +
  % 2L_A\rho_B^{\frac{1}{2}}\|\mathbf{h}_{00}\|+L_A^2)^2 \big]_{d_{00}} \\
  &=\mathbb{E}\big[ \rho_B^2|\mathbf{h}_{00}|^4 \big]_{d_{00}}+
  \mathbb{E}\big[  4 \rho_B^{\frac{3}{2}}L_A\|\mathbf{h}_{00}\|^3  \big]_{d_{00}} \\ &+
  \mathbb{E}\big[ 6\rho_B L_A^2  |\mathbf{h}_{00}|^2\big]_{d_{00}}
   +
  \mathbb{E}\big[ 4 \rho_B^{\frac{1}{2}}L_A^3 \|\mathbf{h}_{00}\|\big]_{d_{00}} +
  L_A^4 
  \\ &\overset{(a)}{=}\rho_B^2 N_B(N_B+1)\beta_{00}^2
  +4\rho_B^{\frac{3}{2}}L_A \frac{\Gamma(N_B +\frac{3}{2})}{\Gamma(N_B)}\beta_{00}^{\frac{3}{2}}
  \\ & +6\rho_B L_A^2 N_B \beta_{00}
  +4\rho_B^{\frac{1}{2}}L_A^3\frac{\Gamma(N_B +\frac{1}{2})}{\Gamma(N_B)}\beta_{00}^{\frac{1}{2}}+ L_A^4  ,
\end{aligned}
\end{equation}
where the relation (a) comes from the derivation of the raw moments of Gamma distribution.
Denote $X$ as a Gamma random variable $X\sim \Gamma(a,\theta)$, its
raw moment is 
\begin{equation}\label{eq:gammaraw}
  \begin{aligned}
  \mathbb{E}[X^K]
  % &=\int_{x=0}^{\infty}x^k f_X(x)\mathrm{d} x 
  %  =\int_{x=0}^{\infty}\frac{x^{a+k-1}}{\Gamma(a)\theta^a}e^{-\frac{x}{\theta}} \mathrm{d}x \\
  & =\frac{\Gamma(a+k)\theta^k}{\Gamma(a)}
  \int_{x=0}^{\infty}\frac{x^{a+k-1}}{\Gamma(a+k)\theta^{a+k}}e^{-\frac{x}{\theta}} \mathrm{d}x
  \\ &\overset{(a)}{=}\frac{\Gamma(a+k)\theta^k}{\Gamma(a)}, \forall k+a>0,
\end{aligned}
\end{equation} 
where relation (a) comes from that  the integral of the PDF of a Gamma random variable is 
equal to 1.

This completes the derivation of $ \mathbb{E}[S_0]$ and $ \mathbb{E}[S_0^2]$.

\section{Coverage Probability Conversion} \label{Appendix:coverage}
 
\subsection{Transformation of the Coverage Probability}

The CCDF of the Gamma random variable   $X \sim \Gamma(k,\theta)$ is 
\begin{equation}\label{eq:gammaCCDF}
  \begin{aligned}
    \mathbb{P}[X>y]=1-F(y;k,\theta)=\sum_{i=0}^{k-1} \frac{1}{i!} 
    \big( \frac{y}{\theta}\big)^i e^{-\frac{y}{\theta}}.
  \end{aligned}
\end{equation} 

Therefore, the  coverage probability   with  respect to $d_{00}$ is

\begin{equation}
  \begin{aligned}
    &\mathbb{P}[ S_{0} > T(I_{B0}+I_{B}+\bar{I}_{A}+\sigma^2)]\\
    &=\mathbb{E}_{I_{B0},I_{B}}\Big[\sum_{i=0}^{k_{S_0}-1} \frac{1}{i!} 
    \big( \frac{T(I_{B0}+I_{B}+ I_{e})}{\theta_{S_0}}\big)^i 
    e^{-\frac{T(I_{B0}+I_{B}+ I_{e} )}{\theta_{S_0}}}\Big] \\
    % &\overset{(a)}{=}\sum_{i=0}^{k_{S_0}-1} \frac{1}{i!} 
    % \mathbb{E}_{I_{B0},I_{B}}\Big[ 
    % \big( \frac{T(I_{B0}+I_{B}+ I_{e})}{\theta_{S_0}}\big)^i 
    % e^{-\frac{T(I_{B0}+I_{B}+ I_{e} )}{\theta_{S_0}}}\Big] \\
    & \overset{(a)}{=}\sum_{i=0}^{k_{S_0}-1} \frac{(-1)^i}{i!} \frac{\partial^i}{\partial^i s}
     \Big\{\mathbb{E}_{I_{B0},I_{B}}\Big[
    e^{-s\frac{T(I_{B0}+I_{B}+ I_{e} )}{\theta_{S_0}}}\Big]\Big\}_{s=1} \\
    &\overset{(b)}{=}\!\sum_{i=0}^{k_{S_0}\!-\!1} \!\frac{(-1)^i}{i!} 
    \!\frac{\partial^i}{\partial^i s}\!
    \Big\{\!e^{-s\frac{TI_{e}}{\theta_{S_0}}}
    \!\mathbb{E}_{I_{B0}}\!\Big[\!
    e^{-s\frac{TI_{B0}}{\theta_{S_0}}}\!\Big]
    \!\mathbb{E}_{I_{B}}\!\Big[\!
      e^{-s\frac{TI_{B}}{\theta_{S_0}}}\!\Big]\Big\}_{s=1}
    \\
    & =\sum_{i=0}^{k_{S_0}-1} \frac{(-1)^i}{i!} \frac{\partial^i}{\partial^i s}
    \Big\{ e^{-s\frac{TI_{e}}{\theta_{S_0}}}
      \mathcal{L}_{Y_{I_{B0}}|d_{00}}(s) 
      \mathcal{L}_{Y_{I_{B}}|d_{00}}(s)   \Big\}_{s=1},
  \end{aligned}
\end{equation}
% where relation (a) is obtained by interchanging the expectation operator 
% and the sum operator.
where relation (a) comes from that $\frac{\partial^i (e^{-sY})}{\partial^i s}
=(-Y)^i e^{-sY} $ and  $Y$ is substituted for 
$\frac{T(I_{B0}+I_{B}+ I_{e})}{\theta_{S_0}}$.
Relation (b) can be obtained since $I_{B0} $ and $I_{B}$ are independent of 
each other. 
Finally, the coverage probability is transformed into the form of higher-order 
derivatives of $ \mathcal{L}_{Y_{I_{B0}}|d_{00}}$ and $\mathcal{L}_{Y_{I_{B}}|d_{00}}$.
 
% This completes the proof of Lemma \ref{lemcoverage}.

\subsection{Derivation of the Laplace Transform}
% \section{Proof of Lemma \ref{lemlaplace}}\label{Appendix:Ylap}
Before the calculation of the Laplace transforms of $Y_{I_{B0}}$ and 
$Y_{I_{B}}$, the moment generating function of the Gamma variable 
is first introduced. For any Gamma variable $X\sim \Gamma(k,\theta)$, the
corresponding moment generating function for $ t<\frac{1}{\theta}$ is 
\begin{equation}\label{eq:mgf}
  M_X (t)=\mathbb{E}[e^{tX}]=\int_{0}^{\infty}e^{tx}f_X(x)\mathrm{d}x=\big(1-t\theta \big)^{-k}.
\end{equation}

According to (\ref{eq:IB0-1}), the intra-cell interference $I_{B0}$ follows the 
Gamma distribution with respect to $d_{00}$ as 
$I_{B0}=P_B\eta_B \beta_{00}\kappa_{B,0}$, where $\kappa_{B,0}\sim 
\Gamma( |\bar{\phi}_{B}|-1,1) $. Therefore, the Laplace transform of $Y_{I_{B0}}$
is 
\begin{equation}\label{eq:lapyib0}
   \resizebox{1\hsize}{!}{$
  \begin{aligned}
    \mathcal{L}_{Y_{I_{B0}}|d_{00}}(s) 
    % &= \mathbb{E}_{I_{B0} }\Big[
    %   e^{-s\frac{TI_{B0}}{\theta_{S_0}}}\Big] \\
      & \!=\!\mathbb{E}_{I_{B0}}\Big[
        e^{-s\frac{TP_B\eta_B \beta_{00}\kappa_{B,0}}{\theta_{S_0}}}\Big]  
        \!\overset{(a)}{=}\!\Big(1\!+\!s\frac{T\rho_B \beta_{00}}{\theta_{S_0}}
          \Big)^{1\!-\!|\bar{\phi}_{B}|}.
  \end{aligned}$}
\end{equation}

Based on (\ref{eq:IBkappa}), the inter-cell interference $I_{B}$ is rewritten as 
the sum of Gamma variables. Since the expectation is taken over   $\Lambda_B$
and the random channel fading with the nearest BS distance $d_{00}$, 
the Laplace transform of 
$Y_{I_{B}}$ is
\begin{equation}
  \begin{aligned}\label{eq:lapyib}
    &\mathcal{L}_{Y_{I_{B}}|d_{00}}(s) 
    % = \mathbb{E}_{I_{B} }\Big[
    %   e^{-s\frac{TI_{B}}{\theta_{S_0}}}\Big] \\
    %   & 
      =
      \mathbb{E}_{I_{B }}\Big[
        e^{-s\frac{T \rho_B\sum\limits_{m\in \omega_B\backslash\{0\}}\beta_{m0} \kappa_{B,m0}}{\theta_{S_0}}}\Big]\\
  & =  
  \mathbb{E}_{\Lambda_{B}}\bigg[ \prod_{m\in \omega_B\backslash\{0\}}
   \mathbb{E}_{\kappa_{B,m0}}\Big[
    e^{-s\frac{T \rho_B \beta_{m0} \kappa_{B,m0}}{\theta_{S_0}}}\Big]\bigg]\\
    &\overset{(a)}{=} 
    \mathbb{E}_{\Lambda_{B}}\bigg[ \prod_{m\in \omega_B\backslash\{0\}}  
    \big(1+s\frac{T \rho_B \beta_{m0} }{\theta_{S_0}} \big)^{- |\bar{\phi}_{B}|}\bigg]\\
    &\overset{(b)}{=} 
    \mathrm{exp}\Big( 2\pi \lambda_B \int_{d_{00}}^{\sqrt{\frac{|\mathcal{A}|}{\pi}}} \big[
      \big(1+s\frac{T \rho_B \beta_{0} r^{-\alpha_1} }{\theta_{S_0}} \big)^{- |\bar{\phi}_{B}|}-1
      \big] r\mathrm{d}r   \Big), \\   
  \end{aligned}
\end{equation}
where $\beta_{m0}=\beta_0 d_{m0}^{-\alpha_1} $, and $\kappa_{B,m0}$ follows the i.i.d. 
$\Gamma(|\bar{\phi}_{B}|,1), \forall m \in\omega_B$.
Relation (a) comes from the moment generating function  
in (\ref{eq:mgf}).
Relation (b) comes from the probability
generating functional (PGFL) of 
a  homogeneous PPP that converts the random product 
over the PP domain to an integral as \cite{haenggi2012stochastic} 
\begin{equation}
  \mathbb{E}\big[\prod_{x\in \Lambda} f(x) \big]
  =\mathrm{exp} \big[ -\int_{\mathbb{R}^2} \lambda(x)[1-f(x)]\mathrm{d}x \big],
\end{equation}
where $\lambda(x)$ denotes the density of PPP $\Lambda$.
% $\phi$ with intensity $\lambda$, which is given by 
% \begin{equation}
%   \mathbb{E} \Big[ \prod_{x\in\phi  } f(x)  \Big]= \mathrm{exp}  \Big( \lambda
%   \int_{\mathbb{R}^2} \big(f(x)-1\big)\mathrm{d}x\Big).
% \end{equation}

This completes the conversion of coverage probability.

\section{Coverage Probability with Large $k_{S_0}$}\label{Appendix:conditionalp}

% According to the analysis in Section III-C and Section IV-A, for any $d_{00}$, 
% $I_{B0}$ and $I_{B}$ are independent of each other, while $I_{B}^{\star}= I_{B0} +I_{B}$
% can be regarded as a sum of  independent Gamma variables. 
% According to the second-order moment matching technique for Gamma distributions, 
% for the sum of independent Gamma random variables $\{x_{m}\}$ with $x_{m}\sim 
% \Gamma (k_{m},\theta_{m})$, it can be approximated as a new Gamma random variable 
% whose shape and scale parameters are $k=(\sum_m k_m \theta_m)^2/ \sum_m k_m \theta_m^2$,
% $\theta= \sum_m k_m \theta_m^2/\sum_m k_m \theta_m$, respectively.
% Therefore, $I_{B}^{\star}$ is approximated accordingly as 
% $ I_{B}^{\star}\sim \Gamma(k_{I_{B }^{\star}} , \theta_{I_{B }^{\star}})$,
% where
%  \begin{equation}
%   k_{I_{B }^{\star}}=
%   \frac{\big(\sum_{m\in \omega_B} |\bar{\phi}_{B}|\beta_{m0}-\beta_{00} \big)^2 }
%   {\sum_{m\in \omega_B} |\bar{\phi}_{B}|\beta_{m0}^2-\beta_{00}^2},
%  \end{equation}
% \begin{equation}
%   \theta_{I_{B }^{\star}}=
%   \frac{\sum_{m\in \omega_B} |\bar{\phi}_{B}|\rho_{B}\beta_{m0}^2-\rho_{B}\beta_{00}^2}
%   {\sum_{m\in \omega_B} |\bar{\phi}_{B}|\beta_{m0}-\beta_{00}}.
% \end{equation}

Considering that the coverage probability changes according to the value of
$I_{e}^{\star}=\frac{\bar{S}_0}{T}-\bar{I}_{A}-\sigma^2$, 
the following 
discussion will be categorized according to the 
different case of $I_{e }^{\star}$.

\subsection{The case of $I_{e }^{\star}\leq 0 $}

Since  $ I_{B0} +I_{B}$ is positive,  
the coverage probability with $d_{00}$ is  
\begin{equation}
  \begin{aligned}
    p_{\mathrm{c}}(d_{00})  
     =\mathbb{P}[ I_{B0} +I_{B}< I_{e}^{\star}] 
    =0, \quad I_{e}^{\star} \leq 0.
  \end{aligned}
\end{equation}

\subsection{The case of $I_{e }^{\star}> 0 $}

% With $I_{e }^{\star}< 0 $, the distribution probability of 
% $I_{\mathrm{sum}}$ is 
% \begin{equation}
%   \mathbb{P}[   I_{\mathrm{sum}}\leq 0]=
%   \mathbb{P}[  I_{B0} +I_{B}\leq \frac{L_A^2}{T}-\bar{I}_{A}-\sigma^2].
% \end{equation}

The coverage probability is transformed to the CDF
of $I_{B0} + I_{B}$.
Note that 
$I_{B0}$ and $I_{B}$ are independent of each other for any given $d_{00}$. 
Therefore, the coverage probability can be decomposed in the form of an 
integral with respect to $I_{B0}$ as 
\begin{equation}\label{eq:pie>0}
  \begin{aligned}
  p_c(d_{00}) &=
  \int_{0}^{I_{e}^{\star}} \!
  f_{I_{B0}}(y)
   P[I_B<I_{e }^{\star}-I_{B0}|I_{B0}=y] \mathrm{d} y
   \\ &=
   \int_{0}^{I_{e}^{\star}} \!
  f_{I_{B0}}(y)
  F_{I_{B}}(I_{e }^{\star}-y) \mathrm{d} y,I_{e }^{\star}> 0 ,
  %  \\ &=\int_{0}^{I_{e}^{\star}} \!
  %  f_{I_{B0}}(y)
  %  \mathcal{L}^{-1}\Big\{ \frac{1}{s} 
  %  \mathcal{L}_{I_B|d_{00}}(s)
  %   \Big\}(I_{e}^{\star}\!-\!y) \mathrm{d} y    
  \end{aligned}
\end{equation}

where  $f_{I_{B0}}(y)$  is   the PDF of Gamma random variables as
\begin{equation}
  f_{I_{B0}}(y)=\frac{1}{\Gamma(|\bar{\phi}_{B}|-1) 
  (\rho_B \beta_{00})^{|\bar{\phi}_{B}|-1}  }y^{|\bar{\phi}_{B}|-2}
  e^{\frac{-y}{\rho_B \beta_{00}}}, y>0.
\end{equation} 

With the PGFL  of the PPP $\Lambda_B$
and a similar derivation as in Lemma \ref{lemcoverage}, 
 the Laplace transform of 
$I_B$ with given $d_{00}$ is 
\begin{equation}
   \resizebox{1\hsize}{!}{$
  \begin{aligned}\label{eq:lapib}
    &\mathcal{L}_{I_{B}|d_{00}}(s) 
    \overset{(a)}{=}
      \mathbb{E}_{I_{B}}\Big[
        e^{-s \rho_B\sum\limits_{m\in \omega_B\backslash\{0\}}\beta_{m0} \kappa_{B,m0}}\Big]\\
        \\ &=\!\mathbb{E}\Big[\!
          \prod_{m\in\omega_B \backslash \{0\}}\! e^{-s \rho_B  \beta_{m0}\kappa_{B,m0} } \Big]
          \!\overset{(b)}{=}\!\mathbb{E}\Big[
            \!\prod_{m\in\omega_B \backslash \{0\}} \!(1\!+\!s\rho_B\beta_{m0})^{-|\bar{\phi_B}|}\Big]\\
          &\overset{(c)}{=}\mathrm{exp}\Big\{  2\pi \lambda_B
          \int_{d_{00}}^{\sqrt{\frac{|\mathcal{A}|}{\pi}}} 
          \big[ (1+s\rho_B\beta_{0 }r^{-\alpha_1})^{-|\bar{\phi}_B|}-1  \big]r \mathrm{d}r\Big\} ,
  \end{aligned}$}
\end{equation}
where relation (a) comes from the approximation in (\ref{eq:IBkappa}).
Relation (b) holds because $\kappa_{B,m0}$ and $\kappa_{B,n0}, \forall m\neq n $ 
are independent of each other and for any Gamma variable, 
its moment generating function is shown in (\ref{eq:mgf}).
Relation (c) is obtained from the PGFL of PPP $\Lambda_B$.

For the Laplace transform of   $I_{B}$, the following
Gauss hypergeometric function $_2F_1$ is introduced that \cite{olver2010nist}
\begin{equation}
  \begin{aligned}
    G(r)&\triangleq\int
     \big[ (1+s\rho_B\beta_{0 }r^{-\alpha_1})^{-|\bar{\phi}_B|}-1  \big]r \mathrm{d}r
    \\ &=\frac{1}{2}r^2\big[ {_{2}F_1}(-\frac{2}{\alpha_1},|\bar{\phi}_B|;\frac{\alpha_1-2}{\alpha_1},-r^{-\alpha_1}\beta_0\rho_Bs)\big]-1 
   .
  \end{aligned}
 \end{equation}
Therefore, the Laplace transform of  $I_{B}$ is further written as 
\begin{equation}
  \mathcal{L}_{I_{B}|d_{00}}(s) =\mathrm{exp}\big\{
    2\pi \lambda_B [G(\sqrt{\frac{|\mathcal{A}|}{\pi}})-
    G(d_{00})
     ] \big\}.
\end{equation}

Further, the CDF $F_{I_{B}}(x)$ of the interference power $I_B$  
is obtained by taking the inverse Laplace transform of (\ref{eq:lapib}) as
\begin{equation}
  \begin{aligned}
    F_{I_{B}}(x)  & =
   \mathcal{L}^{-1}\Big\{ \frac{1}{s} 
   \mathcal{L}_{I_B|d_{00}}(s)
    \Big\}(x).
  \end{aligned}
\end{equation}

After substituting $(I_{e}^{\star} - y)$ for $x$, the inverse Laplace transform 
can be computed using Matlab. This completes the analysis of coverage probability with large $k_{S_0}$.

\section{Derivation of $ \mathbb{E}[S_I]$ and $ \mathbb{E}[S_I^2]$}\label{Appendix:SI}

\subsection{Derivation of $\mathbb{E}[S_I]$}
With given $d_{00}$ and Lemma \ref{lemassLA}, the first moment $\mathbb{E}[S_I]$ 
is  
\begin{equation}
  \begin{aligned}
    &\mathbb{E}[S_I]= \mathbb{E}[S_0]+\mathbb{E}[I_{B0}]\\
    & \overset{(a)}{=}\rho_B  \beta_{00} ( N_B +|\bar{\phi}_{B}|-1)
    +2\rho_B^{\frac{1}{2}}\beta_{00}^{\frac{1}{2}}L_A\frac{\Gamma(N_B +\frac{1}{2})}{\Gamma(N_B)}  
    +L_{A}^{2}  ,
  \end{aligned}
\end{equation}
where relation (a) is obtained from (\ref{eq:ES_0}) and (\ref{eq:IB0-1}).

\subsection{Derivation of $\mathbb{E}[S_I^2]$}
For the second  moment $\mathbb{E}[S_I^2]$  in condition of $d_{00}$, there is
\begin{equation}
  \mathbb{E}[S_I^2]= \mathbb{E}[S_0^2]+\mathbb{E}[I_{B0}^2]+2\mathbb{E}[I_{B0}S_0 ],
\end{equation}
where $\mathbb{E}[S_0^2]$ is derived in (\ref{eq:ES_02}). 

By calculating the raw moments of Gamma random variable $\kappa_{B,0}$ with (\ref{eq:gammaraw}), $\mathbb{E}[I_{B0}^2]$
is derived as 
\begin{equation}\label{eq:EIB02}
  \begin{aligned}
    \mathbb{E}[I_{B0}^2]&=(\rho_B   \beta_{00}\kappa_{B,0})^2 
      =\rho_B^2   \beta_{00}^2 \frac{\Gamma(|\bar{\phi}_{B}|+1)}{\Gamma(|\bar{\phi}_{B}|-1)}\\
    &=\rho_B^2   \beta_{00}^2 |\bar{\phi}_{B}|(|\bar{\phi}_{B}|-1).
  \end{aligned}
\end{equation}

Further, the third part $\mathbb{E}[I_{B0}S_0 ]$ is approximated as 

\begin{equation}\label{eq:EIB0S0}
  \resizebox{1\hsize}{!}{$
  \begin{aligned}
    \mathbb{E}[I_{B0}S_0 ]&=\mathbb{E}[\rho_B \!\sum_{n\in \phi_{B,0}\backslash \{0\}} 
    \Big| \mathbf{h}_{00}^H \frac{\mathbf{h}_{0n}}{\|\mathbf{h}_{0n} \|}\Big|^2 
    ( \sqrt{\rho_B}\|\mathbf{h}_{00} \|  \! +\!L_A )^2] \\
  & = \mathbb{E}[ 
      \underbrace{\rho_B^2\sum_{n\in \phi_{B,0}\backslash \{0\}} 
      \Big| \|\mathbf{h}_{00} \| \mathbf{h}_{00}^H \frac{\mathbf{h}_{0n}}{\|\mathbf{h}_{0n} \|}\Big|^2   }_{A}
      \\ &+ \underbrace{2\rho_B^\frac{3}{2} L_A \sum_{n\in \phi_{B,0}\backslash \{0\}} 
      \Big| \sqrt{\|\mathbf{h}_{00} \|} \mathbf{h}_{00}^H \frac{\mathbf{h}_{0n}}{\|\mathbf{h}_{0n} \|}\Big|^2 }_B
      \\ &+ \underbrace{\rho_B L_A^2 \sum_{n\in \phi_{B,0}\backslash \{0\}} 
      \Big| \mathbf{h}_{00}^H \frac{\mathbf{h}_{0n}}{\|\mathbf{h}_{0n} \|}\Big|^2 }_C ].
  \end{aligned}$}
\end{equation}

The vector $\|\mathbf{h}_{00} \| \mathbf{h}_{00}^H$ in term $A$ is denoted as 
$\mathbf{f}_{A}^{H}=[f^{'}_1,...,f^{'}_{N_B}]$. Then the corresponding $\mathbb{E}[\mathbf{f}_{A}
\mathbf{f}_{A}^{H}]$ is 
\begin{equation}
  \mathbb{E}[\mathbf{f}_{A}
\mathbf{f}_{A}^{H}]\!=\!\mathbb{E}\begin{bmatrix}
  |\mathbf{h}_{00} |^2  |h_{00,1}|^2 & \!\cdots\!&|\mathbf{h}_{00} |^2 h_{00,1}^{ H}h_{00,N_B} \\
  \vdots & \!\ddots\! &  \vdots\\
  |\mathbf{h}_{00} |^2h_{00,N_B}^{ H}h_{00,1}&\!\cdots\! &|\mathbf{h}_{00} |^2 |h_{00,N_B}|^2
  \end{bmatrix},
\end{equation}
where   the diagonal and  non-diagonal elements of $ \mathbb{E}[\mathbf{f}_{A}
\mathbf{f}_{A}^{H}]$ are
\begin{equation}\label{eq:Adiag}
  \begin{aligned}
    \mathbb{E}[   |\mathbf{h}_{00} |^2  |h_{00,i}|^2 ]&=
    \mathbb{E}[   |h_{00,i}|^2\sum_{j\neq i} |h_{00,j}|^2 ]+
    \mathbb{E}[   |h_{00,i}|^4 ] \\
    &= (N_B-1)\beta_{00}^2+\frac{\Gamma(3)}{\Gamma(1)}\beta_{00}^2=(N_B+1)\beta_{00}^2,
  \end{aligned}  
\end{equation}
\begin{equation}\label{eq:Anondiag}
  \begin{aligned}
    \mathbb{E}[  |\mathbf{h}_{00} |^2 h_{00,i}^{ H}h_{00,j}  ]&=
    \mathbb{E}[ h_{00,i}^{ H}h_{00,j}   \sum_{k\neq i,j} |h_{00,k}|^2 ]\\ 
    &+
    \mathbb{E}[ h_{00,i}^{ H}h_{00,j}(|h_{00,i}|^2+|h_{00,j}|^2) ]=0.
  \end{aligned}  
\end{equation}

From (\ref{eq:Adiag}) and (\ref{eq:Anondiag}), $\mathbf{f}_{A}$ is an
isotropic vector with $ \mathbb{E}[\mathbf{f}_{A}
\mathbf{f}_{A}^{H}]=(N_B+1)\beta_{00}^2 \mathbf{I}_{N_B} $. 
Based on  Lemma \ref{lem3}, $\forall n\in \phi_{B,0}\backslash \{0\} $,
 the power distribution of its projection
$\Big| \mathbf{f}_{A}^{H} \frac{\mathbf{h}_{0n}}{\|\mathbf{h}_{0n} \|}\Big|^2$ in term $A $
follows the Gamma distribution $\Gamma(1,(N_B+1)\beta_{00}^2)$. 
Therefore, the corresponding expectation of term $A $ is 
\begin{equation}
  \mathbb{E}[A]=\rho_B^2 (|\bar{\phi}_B|-1)(N_B+1)\beta_{00}^2.
\end{equation}

Then denote   $\sqrt{\|\mathbf{h}_{00} \|} \mathbf{h}_{00}^H $ in term $B$ 
as $\mathbf{f}_{B}^{H}$. There is 
\begin{equation}
  \mathbb{E}[\mathbf{f}_{B}
\mathbf{f}_{B}^{H}]\!=\!\mathbb{E}\begin{bmatrix}
  \|\mathbf{h}_{00} \|   |h_{00,1}|^2 \!& \!\cdots\!&\|\mathbf{h}_{00} \| h_{00,1}^{ H}h_{00,N_B} \\
  \vdots\! & \!\ddots\! &  \vdots\\
  \|\mathbf{h}_{00} \|h_{00,N_B}^{ H}h_{00,1}\!&\!\cdots \!&\|\mathbf{h}_{00} \| |h_{00,N_B}|^2
  \end{bmatrix}.
\end{equation}

The diagonal and  non-diagonal elements of $ \mathbb{E}[\mathbf{f}_{B}
\mathbf{f}_{B}^{H}]$ can be expressed respectively as 
\begin{equation}\label{eq:Bdiag}
  \begin{aligned}
    \mathbb{E}[   \|\mathbf{h}_{00} \| |h_{00,i}|^2 ]&=
   \frac{1}{N_B} \mathbb{E}[   \|\mathbf{h}_{00} \|\sum_{i=1,...,N_B} |h_{00,i}|^2 ]  \\
    & =\frac{1}{N_B} \mathbb{E}[   \|\mathbf{h}_{00}\|^3]=\frac{1}{N_B}
    \frac{\Gamma(N_B+\frac{3}{2})}{\Gamma(N_B)} \beta_{00}^{\frac{3}{2}},
  \end{aligned}  
\end{equation}
\begin{equation}\label{eq:Bnondiag}
  \begin{aligned}
    \mathbb{E}[  \|\mathbf{h}_{00} \| h_{00,i}^{ H}h_{00,j}  ]&=
    \mathbb{E}[ h_{00,i}^{ H}h_{00,j}   \sum_{k=1,...,N_B} |h_{00,k}|^2 ] =0.
  \end{aligned}  
\end{equation}

Similarly, $\mathbf{f}_{B}$ is an
isotropic vector with $ \mathbb{E}[\mathbf{f}_{B}
\mathbf{f}_{B}^{H}]=\frac{1}{N_B}
\frac{\Gamma(N_B+\frac{3}{2})}{\Gamma(N_B)} \beta_{00}^{\frac{3}{2}} \mathbf{I}_{N_B} $.
Thus, $\forall n\in \phi_{B,0}\backslash \{0\} $, there is the power distribution
 $\Big| \mathbf{f}_{B}^{H} \frac{\mathbf{h}_{0n}}{\|\mathbf{h}_{0n} \|}\Big|^2
\sim \Gamma(1,\frac{1}{N_B}
\frac{\Gamma(N_B+\frac{3}{2})}{\Gamma(N_B)} \beta_{00}^{\frac{3}{2}})$ for the term $B$ 
based on  Lemma \ref{lem3}, and the expectation of term $B $ is 
\begin{equation}
  \mathbb{E}[B]=2\rho_B^\frac{3}{2} L_A (|\bar{\phi}_B|-1)
  \frac{1}{N_B}
\frac{\Gamma(N_B+\frac{3}{2})}{\Gamma(N_B)} \beta_{00}^{\frac{3}{2}}.
\end{equation}

%For the term $C$, there is $\Big| \mathbf{h}_{00}^H \frac{\mathbf{h}_{0n}}{\|\mathbf{h}_{0n} \|}\Big|^2
%\sim \Gamma(1,\beta_{00})$.    The expectation of term $C $ is

For the term $C$, with $\Big| \mathbf{h}_{00}^H \frac{\mathbf{h}_{0n}}{\|\mathbf{h}_{0n} \|}\Big|^2
\sim \Gamma(1,\beta_{00})$, there is
\begin{equation}
  \mathbb{E}[C]=\rho_B L_A^2 (|\bar{\phi}_B|-1)\beta_{00}.
\end{equation} 

Finally, the second-order moment $\mathbb{E}[S_I^2]$ of $S_I$ in condition of $d_{00}$
is obtained with (\ref{eq:ES_02}), (\ref{eq:EIB02}) and (\ref{eq:EIB0S0}) as 
\begin{equation}
  \begin{aligned}
  \mathbb{E}[S_I^2]
%   &=\rho_B^2 N_B(N_B+1)\beta_{00}^2
%   +4\rho_B^{\frac{3}{2}}L_A \frac{\Gamma(N_B +\frac{3}{2})}{\Gamma(N_B)}\beta_{00}^{\frac{3}{2}}
%   \\ & +6\rho_B L_A^2 N_B \beta_{00}
%   +4\rho_B^{\frac{1}{2}}L_A^3\frac{\Gamma(N_B +\frac{1}{2})}{\Gamma(N_B)}\beta_{00}^{\frac{1}{2}}+L_{A}^{4} \\
%   &+\rho_B^2   \beta_{00}^2 |\bar{\phi}_{B}|(|\bar{\phi}_{B}|-1) \\
%   &+    \rho_B^2 (|\bar{\phi}_B|-1)(N_B+1)\beta_{00}^2\\
%   &+2\rho_B^\frac{3}{2} L_A (|\bar{\phi}_B|-1)
%   \frac{1}{N_B}
% \frac{\Gamma(N_B+\frac{3}{2})}{\Gamma(N_B)} \beta_{00}^{\frac{3}{2}} \\
% &+\rho_B L_A^2 (|\bar{\phi}_B|-1)\beta_{00} \\
& =\rho_B^2  \beta_{00}^2((N_B+|\bar{\phi}_B|)^2+|\bar{\phi}_B|-N_B-2)\\ 
&+ 4\rho_B^\frac{3}{2}\beta_{00}^{\frac{3}{2}} L_A
\frac{\Gamma(N_B +\frac{3}{2})}{\Gamma(N_B)}(1+\frac{(|\bar{\phi}_B|-1)}{N_B}) \\
& +\rho_B\beta_{00}L_A^2 (6N_B+2|\bar{\phi}_B|-2)\\
&+4\rho_B^{\frac{1}{2}}L_A^3\frac{\Gamma(N_B +\frac{1}{2})}{\Gamma(N_B)}\beta_{00}^{\frac{1}{2}}+L_{A}^{4}.
\end{aligned}
\end{equation}

This completes the derivation of $ \mathbb{E}[S_I]$ and $ \mathbb{E}[S_I^2]$.

\end{appendices}

%\begin{thebibliography}{1}

\bibliographystyle{IEEEtran}

\bibliography{IEEEabrv,ref}

%\end{thebibliography}

\end{document}